\documentclass{article}
\usepackage{graphicx} % Required for inserting images
\usepackage{amsfonts}
\usepackage{amsmath}
\usepackage{amsthm}
\usepackage{mathrsfs}
\usepackage{color}
\usepackage{subcaption}
\usepackage{float}
\usepackage{bbm}
\newtheorem{definition}{Definition}
\newtheorem{theorem}{Theorem}
\newtheorem{lemma}{Lemma}
\newtheorem{corollary}{Corollary}
\newtheorem{proposition}{Proposition}
\newtheorem{example}{Example}
\newtheorem{remark}{Remark}

\usepackage{amsmath}
\usepackage{amssymb}
\usepackage{mathtools}
\usepackage{amsthm}
\usepackage{xcolor}
\usepackage{tcolorbox}
\usepackage{booktabs}
\usepackage[bbgreekl]{mathbbol}
\usepackage{bm}
\newif\ifdraft
\draftfalse
\def\ale#1{\ifdraft\textcolor{magenta}{#1}\else{#1}\fi}
\DeclareSymbolFontAlphabet{\mathbbl}{bbold}
\DeclareMathOperator*{\argmin}{arg\,min}

\ifdraft
\usepackage{lipsum}
\usepackage[notref,notcite,color]{showkeys}
\usepackage[cam,width=10in,height=12in,center]{crop}
\definecolor{refkey}{rgb}{0,.6,0}
\definecolor{labelkey}{cmyk}{0, 1, 0, 0}
\tcbset{colframe=white,colback=magenta!20}
\fi

\def\eps{\varepsilon}

\def\neuron{\xi}
\def\pneuron{p}
\def\weight{{\bm{w}}}

\def\pweight{{\bm{p}}}

\def\bbR{\mathbb{R}}

\def\L{{\rm L}}

\usepackage[textsize=tiny]{todonotes}

\title{\bf An Introduction to Cognidynamics\footnote{This paper is related to the invited talk I gave at the Third Conference on Lifelong Learning Agents (CoLLAs 2024) on the 29th of July 2024.}}
\author{Marco Gori \\
University of Siena}
\date{August 2024}

\begin{document}
\maketitle

\begin{abstract}
    This paper gives an introduction to \textit{Cognidynamics}, that is to the dynamics of cognitive systems driven by optimal objectives imposed over time  
    when they interact either with a defined virtual or with a real-world environment.  
    The proposed theory is developed in the general framework of dynamic programming which leads to think of computational laws dictated by classic Hamiltonian equations. 
    Those equations lead to the formulation of a neural propagation scheme in cognitive agents modeled by dynamic neural networks which exhibits locality in both space and time, thus contributing the longstanding debate on biological plausibility of learning algorithms like Backpropagation.
    We interpret the learning process in terms of energy exchange with the environment and show the crucial role of energy dissipation and its links with focus of attention mechanisms and conscious behavior. 
\end{abstract}

\section{Introduction}
% emphasis on  natural laws of learning
The introduction of focus of attention in the Transformer architecture~\cite{NIPS2017_3f5ee243} can likely be regarded as a paradigm-shift in Machine Learning. Interestingly, transformer-based architectures mostly reported superior results compared to  recurrent neural networks, whose architecture may potentially be more adequate for sequential tasks. However, one should bear in mind that the limitation of capturing long-term dependencies by gradient-based learning algorithms were early pointed out about three decades ago~\cite{Bengio_trnn93,
pascanu2013difficulty, hochreiter2001gradient}. A remarkable ingredient to face the problems of gradient vanishing is that of adopting the gating mechanisms proposed in the LSTM architecture (see e.g.~\cite{Hochreiter:97nips} for an early evidence of the effectiveness of the proposal). LSTM architectures have been in fact the reference architecture for challenging experiments in the last decades and, especially, in conjunction with the explosion of Deep Learning. 

This paper proposes a reformulation of learning which relies on full human-based protocols, where machines are expected to conquer cognitive skills from environmental interactions over time. We assume that, at each time instant, data acquired from the environment is processed with the purpose of updating the current internal representation of the environment, and that the agent is not given the privilege of recording the temporal stream. Basically, we assume to process on-line information just like animals and humans without neither store nor access to internal data collections. We formulate learning in the continuous setting of computation, which better reflects most natural processes\footnote{However, it is worth mentioning, that a related formulation can be proposed in the discrete setting of computation without remarkable differences.}. This setting differs remarkably from the usual assumptions of learning on the basis of collections of separated sequences; basically, we do not rely on sequence segmentation, but  process one single sequence which corresponds with the agent life. This corresponds with promoting the role of time, which ordinarily indexes the environmental information. 
Learning is regarded as an optimization problem of a functional risk that arises from the environmental information over the agent life. 
We use the mathematical apparatus of Dynamic Programming and Optimal Control to approach the problem. We show that the distinctive spirit of learning springs out from the need to solve the Hamilton equations by Cauchy initialization, just like in most related problems of Theoretical Physics. Interestingly, the explicit dependence on time of the Hamiltonian, which reflects the interaction with the environment, leads to the problem of devising approximations of the optimal solution that can only be given by imposing boundary conditions\footnote{It is important to bear in mind that one can use boundary conditions only if the overall data collection is recorded, which is exactly what we are excluding in our theory.}. 
This paper shows that minimization of the functional risk with given boundary conditions can be approximated by using Cauchy conditions, which opens the doors to truly on-line computational schemes.
We analyze the system dynamics behind learning in terms of energy exchange. 
The main result is that learning can only take place if we properly introduce appropriate focus of attention mechanisms, that turn out to play a crucial role in order to control the energy accumulation in the agent. The theory suggests that any dynamical model  in the proposed optimization setting must involve additional learning parameters whose purpose is that of stabilizing Hamiltonian neural propagation and expose them gradually to the environment. This is a fundamental  difference with respect to classic machine learning schemes, since this also corresponds with evidence  from developmental psychology~\cite{Piaget1958,Karplus1980}. 
The term \textit{Cognidynamics} that is introduced in this paper refers to a theory for describing the neural propagation scheme emerging from the described optimization process that is required to take place on-line by fully matching the spirit of the classic citation by Danish theologian, philosopher, and poet  S{\o}ren Kierkegaard: 
``Life can only be understood backwards; but it must be lived forwards.'' 
Interestingly, computers can also somewhat ``live backwards'' due to the accumulation and interaction with huge data collections, which is in fact one of the reasons of the recent spectacular results obtained by Large Language Models. However, a theory of learning for interpreting live - or virtual live - going forward without collecting data is, in itself, a great scientific challenge.
Moreover, it can open the doors to a truly orthogonal technological direction where the emphasis shifts back from cloud computing and huge data collections to thin personal computer and industrial devices that can create a society of intelligent agents which are fully using their computational resources instead of simply acting as lazy clients who are mostly involved simply in handling communication processes.  

\section{Collectionless AI}
The big picture of Artificial Intelligence that emerges from
by Russel and Norvig~\cite{russel2021} is centered 
around a few classic topics, whose methodologies can, amongst others, be characterized
by the noticeable difference that while ``symbolic AI'' is mostly collectionless,
``sub-symbolic AI'' is currently strongly relying on huge data collections. 
Interestingly, Machine Learning,
Communicating, Perceiving, and Acting relies mostly on Statistical methodologies
whose effectiveness has been dramatically improved in the last 
decade because of the access to huge data collections. This has been in fact likely the most important ingredient of the success of Machine Learning  that has found a comfortable place under the umbrella of Statistics. As such, by and large, scientists have gradually become accustomed to taking for granted the fact that it is necessary to progressively accumulate increasingly large data collections. It is noteworthy that even symbolic approaches to AI are based on relevant collections of information, but in that case they are primarily knowledge bases and there is no data directly collected from the environment. 
When focusing on the difference of information that is stored,
a question naturally arises concerning the possibility of exhibiting intelligent 
behavior only thanks to an appropriate internal representation of knowledge.
Clearly, while the knowledge representation typically enjoys the elegance 
and compactness of logic formalism, the storage of patterns apparently leads to 
the inevitable direction of accumulating big data collections. 
However this is indeed very unlikely to happen in nature.
Animals of all species organize environmental information for their own purposes without collecting the patterns that they acquire every day at every moment of their life.
This leads to believe that there is room for collapsing to the common framework\footnote{Most ideas of Collectionless AI have been introduced in~\cite{collectionless-ai}, from which this section is based.} of ``Collectionless AI'' also for Machine Learning. 
The environmental interaction, including the information coming from  humans 
plays a crucial role in the learning process, as well as the agent-by-agent communication. 
We think of agents that can be managed by edge computing devices, 
without necessarily having access to servers, cloud computing and, 
more generally, to the Internet. 
This requires thinking of new learning protocols where machines 
learn in lifelong manner and are expected to 
conquer cognitive skills in a truly human-like context 
that is characterized by environmental interactions, without storing the information acquired from the environment. 

The former completely depends on data collections (clouds) that were possibly supervised beforehand, and where there is no direct/interactive connection between 
who supervised the data and the learning agent, since all goes through the collections. 
Moreover, learning takes place in (distributed) machines with large computational power 
which supports training procedures which are way shorter than the life of the agent 
after the deployment. Differently, Collectionless AI focuses on environmental 
interactions, where the agent ``lives'' in the same environment of the human, 
with a close interaction. As such, the distinction between learning and test 
that is typically at the core of Statistical Machine Learning 
needs to be integrated with more appropriate methods 
that are expected to rely on ``daily assessments''.  
In this new framework, time plays a crucial role, since the way the 
streamed data evolves affects uniformly learning and inferential processes. 
As a matter of fact, the agent can also exploit its own predictions/actions 
to update its internal state.
However, in the supervised continual learning literature, the notion of ``time'' is frequently neglected, and the agent has not an internal state that depends on its previous behaviour/predictions. Continual reinforcement learning follows a specific structure of the learning problem (typical of reinforcement learning, usually augmented with neural networks), while what we discuss here is at a higher abstraction level. In a sense, the ideas of this paper generalize several lifelong-learning-related notions, giving more emphasis to the role of time, to the importance of interactions, and to avoiding the creation of data collections.

Overall, we propose facing the following challenge:
\begin{quote}
\it
Learning takes place without accumulation of collections in nature.\\ Are Collectionless Machine Learning protocols also feasible for gaining intelligent skills in machines?
\end{quote}
From one side, facing this challenge leads to 
better understanding the nature of computational processes taking place in 
biology. From the other side, facing this challenge leads to develop AI solutions 
that go beyond the risks connected with data centralization. Moreover, this challenge might rise to a new scenario that, on the long run, might lead to novel promising technologies more centered around edge-computing-like device, 
without possibly accessing other resources over the network. 
To some extent the spirit of this paper nicely intersects the 
driving directions given in~\cite{depressed}. We do subscribe the authors'
point of view on the picture they give on the  anxious state of many AI scientists
who feel that are not coping with the current pace of AI advancements.
We claim  that Collectionless AI might open the doors to 
important advances of the discipline both from the scientific and technological side.

\section{Why do we need a theory of Cognidynamics?}
\label{why_cd_sec}
% neuroscience 
Beginning from the  excellent models for the behavior of single neurons~\cite{Hodgkin1952}, the dynamical system hypothesis in neuroscience and cognition has been massively investigated over the past decades. As early as at the beginning of the nineties, Anderson, Pellionisz, and Rosenfeld~\cite{anderson1990neurocomputing} edited a seminal book where, amongst others, an important part was devoted to ``Computation and Neurobiology.''
Most studies that sprang out from those early contributions have been mostly devoted to the interpretation of brain images and to the role of system dynamics into the emergence of cognition. The associated  complexity issues have led to drive the connectionist wave, mostly promoted  by the PDP~\cite{Hinton86a,McClelland86a} in the mid-eighties, towards the massive experimentation of static models trained via Backpropagation and related algorithmic schemes. This research direction has fueled the development of statistical foundations of learning and it has given rise to the systematic exploitation of learning mechanisms that, instead of relying on natural laws, benefit from access to huge data sets. The milestones marked by the explosion of deep learning~\cite{lecun2015deep} and transformers~\cite{vaswani2017attention} have shown a scenario that could hardly be predicted even by the fathers of those  Machine Learning innovations. Interestingly, while the studies in computational mechanisms of neurobiology have been regarding biological plausibility as a major issues, the above mentioned current developments which have given rise to spectacular results do not even deal with physical plausibility, since the static processing over deep neural networks relies upon the underlining assumption of infinite velocity of signal propagation. It is worth mentioning that early criticisms on the interpretation of neural propagation in multilayered architectures were carried out by Francis Crick in a seminal paper on Backpropagation biological plausibility~\cite{Crick89}, an issue which, after decades, is still under the lens of investigations~(see e.g.~\cite{song2024a}). 
On the other hand, when moving towards computer science oriented studies, also the recent resurgence of interest in recurrent neural network models~\cite{Melacci_2024_rnn_survey} is not addressing the longstanding issue of biological plausibility; the only neural propagation schemes which exhibit locality in both time and space  still requires to restrict the system dynamics to local feedback connections, which is in fact a special case that was early early discovered at the end of the eighties (see e.g.~\cite{Gori_ijcnn89}).

The above discussion motivates the study of Cognidynamics, namely the search for neural propagation schemes in dynamic neural networks for capturing the secrets of  learning and inference over time under environmental interactions. The purpose of the study is not restricted to understand the qualitative picture that emerges from system dynamics, but it faces directly the overall challenge of discovering information-based laws that are likely driving biological processes in the brain. The underlying assumption is those laws do not face the ambitious task of capturing the details of biological processes. However, we conjecture that cognitive tasks can be given an abstract interpretation in terms of information-based processing regardless of the specific biological implementation that takes place in nature in different species. No matter which living organism we consider, all cognitive tasks continuously address interactions with the environment as well as
survival, reproduction, growth and development, and energy utilization. All those objectives involve an appropriate adaptation to the environment, which is in fact the essence of learning. A teleological character naturally arises for interpreting those processes which leads to optimization principles related to those of  Theoretical Physics. For instance, the original formulation of Maupertuis \cite{Maupertuis-1744}, \cite{Maupertuis-1746} of the Principle of Least Action was motivated by a metaphysical teleological argument that ``nature acts as simply as possible.'' Today we know that in general the action functional is least only for sufficiently short trajectories. 
In any case, while the formulation of the Least Action Principle has given rise to disputes about its  teleological character, the information-based processing behind particles and complex organisms are remarkably different in terms of objectives. It is clear that the invoked optimization schemes in biological organisms play a stronger role in Biology and in intelligent machines than in Physics. This stimulates the study of optimization theories to capture the essence of those adaptation processes under environmental interactions. 
\iffalse
We propose using the framework of dynamic programming and optimal control for the minimization of functional risks computed over time which describe the quality of the environmental interactions. In general we need to specify the initial conditions of the state and the final conditions of the co-state. In doing so, the boundary problem that characterizes the inferential process required to discover the minimum is well-posed. However, this would require to access to data collections which comprehend the whole agent life
\fi
\iffalse
During an intriguing discussion on this issue, R. Feynman~\cite{Feynman-1963} 
was posing the question ``How does the particle find the right path? ... does it smell the neighboring paths to find out whether or not they have more action?'' He invoked the links with the case of light where if we block the way so that the photons could not test all possible paths then we experiment the phenomenon of diffraction. 
Based on the intuition that something similar could take place for particles in the Quantum Mechanics framework, Feynman proposed that wave interference offers the explanation for the choice of the right path. In particular he introduced the path integral~\cite{Feynman:100771} that is based on weighting the different paths according to the factor $e^{iS/\hbar}$. 
\fi
The theory of Cognidynamics is mostly inspired by those principle of Theoretical Physics for capturing fundamental laws of particles. The framework of cognitive systems suggests stressing the importance of targeting optimal solutions more than stationary points of the action. This motives the choice of shifting from the mathematical apparatus of Theoretical Physics to that of Optimal Control under the framework of Dynamic Programming, while still bearing in mind the idea of conceiving information-based laws of cognition, which somewhat resemble of the spirit of any scientific investigation involving time. While neural computation has been mostly driven by statistical foundations to capture regularities in huge data collections, the driving principle of Cognidynamics is that of capturing regularities over time. While huge data collections have shown unexpected capabilities of learning, one should bear in mind that while massively enabling intelligent agents to interact with the environment, we begin promoting a truly orthogonal path of artificial intelligence. Even though we assume not to store information, as times goes, by the amount of environmental information that can be processed significantly exceeds the information available in any data collection. 

\section{Formulation of Lifelong Learning}
In this section we propose a formulation of Lifelong Learning where the classic notion of time plays the role of protagonist. We introduce a class of intelligent agents that we call NARNIAN: NAtuRe-iNspIred computAtional IntelligeNce agents.
In particular we consider the continuous interpretation of time mostly used in Physics, though an associated discrete-time setting can replace the analysis carried out in this paper without significant differences. The choice of the continuous setting of computation comes from the mentioned objective of exploring laws of cognition following the spirit that drives other temporal laws in Science. As mentioned in Section~\ref{why_cd_sec}, once the \textit{temporal environment} is given, one can think of intelligent agents as dynamical systems that are expected to exhibit a high degree of adaptation capabilities. Moreover, just like any living organism, in addition to continuously react to the environment, intelligent agents in temporal environments can likely benefit from inheriting the capacity of  reproduction, growth and development, and energy utilization. \\
~\\
\emph{\sc Survival and Reproduction}\\
Survival is a slippery topic, but one could easily think of intelligent agents which are either alive or dead.   Reproduction is the mechanism for conquering the state of alive; it is in fact  somewhat connected to the massively experimented transfer learning schemes in artificial neural networks. Intelligent agents can follow classic scheme of evolutionary computation which nicely fits in this temporal environment framework\footnote{Survival and reproduction are not be covered in this paper but, as mentioned, the rich literature on evolutionary computation can be explored naturally in this temporal setting of the environment.}. \\
~\\
\emph{\sc Growth and Development}\\
A fundamental difference from current machine learning models based on access to large collections of data immediately emerges in any lifelong learning scheme that is based on a temporal environment, as normally occurs in nature. Every living organism experiences processes of gradual cognitive growth which in some evolved species, especially humans, also presents rather clearly defined development phases~\cite{Piaget1958}. To some extent, it is as if nature intervenes in a sort of protection against excessive exposure to information. For example, the vision process in newborns undergoes a developmental phase in which visual scenes are strongly blurred, and it takes nearly one year to achieve adult's visual acuity. Are those developmental schemes connected to the specificity of human biology or do they come from more general information-based principles?  This paper provides results to sustain the position that similar developmental schemes are characterized by the appropriate adaptation of specific neural connections that prevent information overloading. This is claimed come from informatio-based principles that are independent of specific biological species. 
\\
~\\
\emph{\sc Energy Utilization}\\
As it will be shown in section~\ref{eb_sec}, the process of learning can be interpreted as energy dissipation in the dynamical systems which is acting as an intelligent agent. Basically, it will be shown that the agent receives environmental energy which is used to change the internal energy by dissipation. Clearly, the agent must prevent from accumulating too much internal energy during learning, which is in fact an important driving heuristics for the system dynamics. The end of any learning process corresponds with ending the dissipation process which suggests in any case to carry out policies for bounding the internal and the dissipation energy. Interestingly, just like for sustaining developmental schemes, any intelligent agent implements policies aimed at controlling the energy exchange with the environment by the appropriate adaptation of specific neural connections.\\
~\\
\emph{\sc Neural Network Architecture and Formulation of Learning}\\
\begin{figure}[H]
	\centering
	\includegraphics[width=12cm]{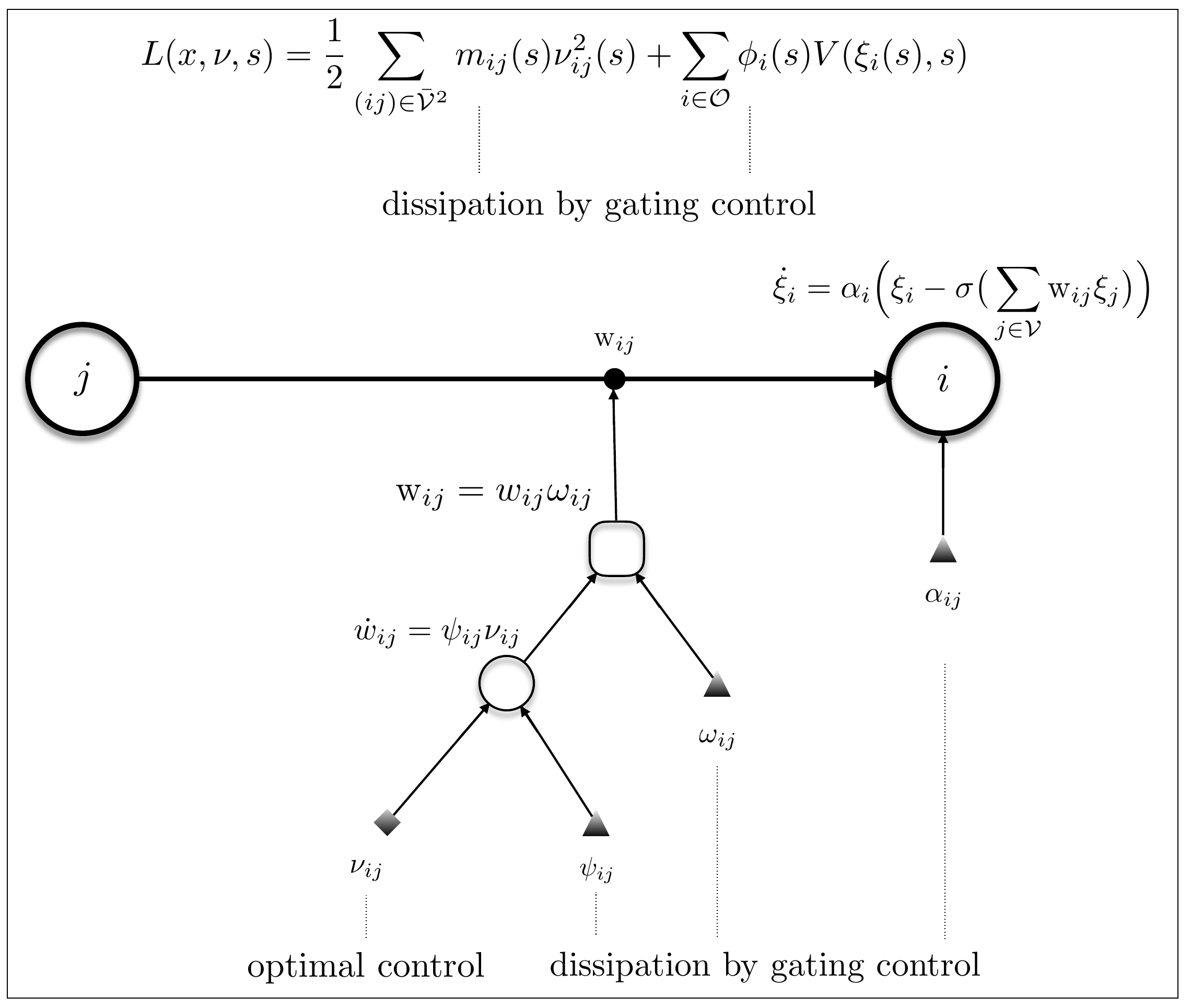}
\caption{\small {\em Neural model, Lagrangian, and the two coordinates of learning. First, the theory of Optimal Control drives the velocity $\nu_{ij}$ of the weights, thus resembling gradient-descent learning policies. Second, the dissipative weights $\zeta$ control the overall learning process by gating mechanisms which takes place on both the network and the Lagrangian.}.
}
\label{NetworkLagrangian-Fig}
\end{figure}
Based on the above premises we consider a dynamical system defined by a continuous-time recurrent neural network, where some of the neurons are reserved to model the interactions with the environment, whereas others are hidden. Basically, the set of vertexes 
${\cal V} ={\cal I} \cup {\cal H} \cup {\cal O}$
contains input $\mathcal{I}$, hidden $\mathcal{H}$, and output units $\mathcal{O}$.
It turns out that for input nodes, $i \in {\cal I}$,
the corresponding neuron returns a value $\xi_i$ that corresponds with the 
input $u_{i}$. Likewise, for output units,
$o \in {\cal O}$, the value $\xi_{o}$ encodes an action that is expected to optimize the agent behavior. Finally, other units $\xi_{h}$, that are characterized by $h \in {\cal H}$, are hidden. 
We also introduce $\bar{\cal V}:= {\cal H} \cup {\cal O}$ which represents the set of non-degenerate neurons. This set is useful also for defining the architecture of the recurrent network which is characterized by the graph ${\cal G} = ({\cal V},{\cal A})$, where 
${\cal A} \subset {\cal V} \times \bar{\cal V}$ is the set of arcs which connect non-degenerate units. The  neural system dynamics obeys the ODE
%\footnote{Throughout the paper we use an extension of Einstein summation
%convention that also applies for equalities: the notation $c= a_{ij}+b_i$ for
%instance it is assumed to stand for $c= sum_{ij} a_{ij}+\sum_i b_i$. Here
%$\sigma(\cdot)$, which is assumed to be a classic sigmoidal function
%(e.g. $\sigma(a) = 1/(1+\exp({-a}))$).}
\begin{equation}
{\cal N}: \ 
\begin{cases}
\neuron_i(t)= u_i(t) & i\in\mathcal{I}\\
a_{i}(t) = \sigma\Bigl(\sum\nolimits_{j \in {\cal V}} {\rm w}_{ij}(t)
\neuron_j(t)\Bigr)& i\in \bar{\cal V}\\
\dot{w}_{ij}(t)=\psi_{ij}(t)\nu_{ij}(t)
& (i,j)\in\mathcal{A}\\
{\rm w}_{ij}(t) = \omega_{ij}(t) w_{ij}(t) & (i,j)\in\mathcal{A}\\
\dot\neuron_i(t)= \alpha_i(t)\Big[-\neuron_i(t)+\sigma\bigl(a_{i}(t) \bigr) \Big]
& i\in \bar{\cal V}
\end{cases}\label{eq:neural-model}
\end{equation}
which holds on the horizon $(0,T)$. 
Here $T>0$ can also be infinite\footnote{For the sake of simplicity, in the remainder of the paper we will sometimes
drop the temporal dependence of the variables.}. 
We denote by $x:= (\neuron,{\rm flatten}(w))$ the overall state which contains the neural activation $\xi$ and the weights associated with ${\cal A}$, while $\sigma(\cdot)$ is a sigmoidal function (e.g. $\sigma(\cdot) = \tanh(\cdot)$).  
The system dynamics of the weights ${\rm w}_{ij}$ is driven by the corresponding velocity ${\nu}_{ij}$, which is properly filtered by $\psi_{ij}$ and $\omega_{ij}$, respectively.  
At any $t$, the value $\alpha_{i}(t)=\tau_{i}^{-1}(t) \geq 0$ can be regarded as inverse of time 
constants $\tau_{i}(t)$ and somewhat characterizes the velocity of performing the associated computation of the activation $a_{i}(t)$. 
\begin{remark} {\bf Network Gating Functions}\\
    $\alpha, \psi$ and $\omega$ are {\em gating functions}.
    Notice that $\alpha$ plays a crucial role for the characterization of the system dynamics. When $\alpha_{i}(t) = 0$ the neuron is {\em inhibited}, whereas it is {\em activated} whenever $\alpha_{i}(t) > 0$. It is worth mentioning that such a fundamental difference in the state {\em (activated/inhibited)} has been systematically observed in  biological neurons. The role of $\psi$ and $\omega$ is that of providing an appropriate gating mechanism on the weights. As it will be shown in remainder of the paper,  they carry out a different type of system dynamics that nicely matches the need of gradual learning and energy control, respectively. We can easily see that $\psi$ can perform {\em time warping}, whereas $\omega$ can carry out {\em pruning strategies}.
\end{remark}
The dynamical system defined by Eq.~\ref{eq:neural-model}, is characterized by the state $x \sim (\xi,w)$ and by $\nu$, which acts as the velocity of the weights. This is in fact the function that will be used for driving the learning process, since it defines the evolution of the weights ${\rm w}$. Notice that, just like {\rm w}, the network gating functions $\zeta_{N} \sim (\alpha,\psi, \omega)$ can also be thought of learning functions. This is in fact a fundamental architectural difference with respect to traditional recurrent neural network models. We will see that an associated learning process takes place concerning the weights $\zeta$.

The interaction with the environment of  the recurrent neural network is expected to minimize the following functional risk
\begin{equation}
R(\nu,T) = \int_{0}^T  \bigg(
        \frac{1}{2} \sum_{i \in \bar{\cal V}}  \sum_{j \in \cal V}  
        m_{ij}(s) \nu_{ij}^{2}(s) + 
         \sum_{i \in {\cal O}} \phi_{i}(s) V(\xi_{i}(s),s) 
    \bigg) \,ds
\label{eq:risk}
\end{equation} 
Here $L(x,\nu,s) = 
\frac{1}{2} \sum_{i \in \bar{\cal V}}  \sum_{j \in \bar{\cal V}}  
m_{ij}(s) \nu_{ij}^{2}(s) + 
\sum_{i \in {\cal O}} \phi_{i}(s) V(\xi_{i}(s),s)$ is the Lagrangian of the paired system~(\ref{eq:neural-model}).
The Lagrangian is characterized by:
\begin{itemize}
    \item The kinetic energy term  $K(s)=\frac{1}{2} m_{ij}(s) \nu_{ij}^{2}(s)$ 
    of the particle $(i,j)$ that is characterized by the
    mass $m_{ij}(s) > 0$  defined over the time interval $[0,T]$. 
    \item The potential energy of field $\xi$, which is in fact 
    related to classic loss functions in Machine Learning. We can express it by $V(\neuron,t) = \mathcal{V}(\neuron,e(t))$, which is indicating the explicit fundamental dependence of this loss term from $e(\cdot)$, which incorporates the information for the environmental interactions, and can be expressed by $e \sim (u,y)$ (inputs and
    outputs). This potential energy is  accumulated over time  according to 
    $\phi_{i}(t)>0$, which a gate function that we also referred to as the
    {\em conscious functions}. The term $\phi_{i}(t)$ enables learning  by involving the  
    corresponding loss term defined by the potential $V(\xi_{t},t)$ (see Fig.~\ref{Conscious-Fig}). 
    Notice that $\phi_{i}(t)$ is attached to all neurons $i \in \bar{\cal V}$, which will be shown to play an important role in the overall learning process.
    \iffalse
    \item Function $\gamma$ takes on values $\gamma_{i}(t) \in \left\{-1,+1\right\}$. Because of the sign assumptions on $m$ and $\phi$, when $\gamma_{i}(t) = -1$ it plays the role of returning a functional which resembles classic action in Mechanics. In particular, in this case, we refer to functional $R(\nu,T)$  as the \textit{Cognitive Action}.
    \fi
    \item The explicit dependence on time in the Lagrangian is useful for carrying out  dissipation mechanisms. Such a dependence is mainly expressing the information coming from the environment by  function $e$ which conveys the information at any time $t$.
\end{itemize}

\begin{figure}[H]
	\centering
	\includegraphics[width=11cm]{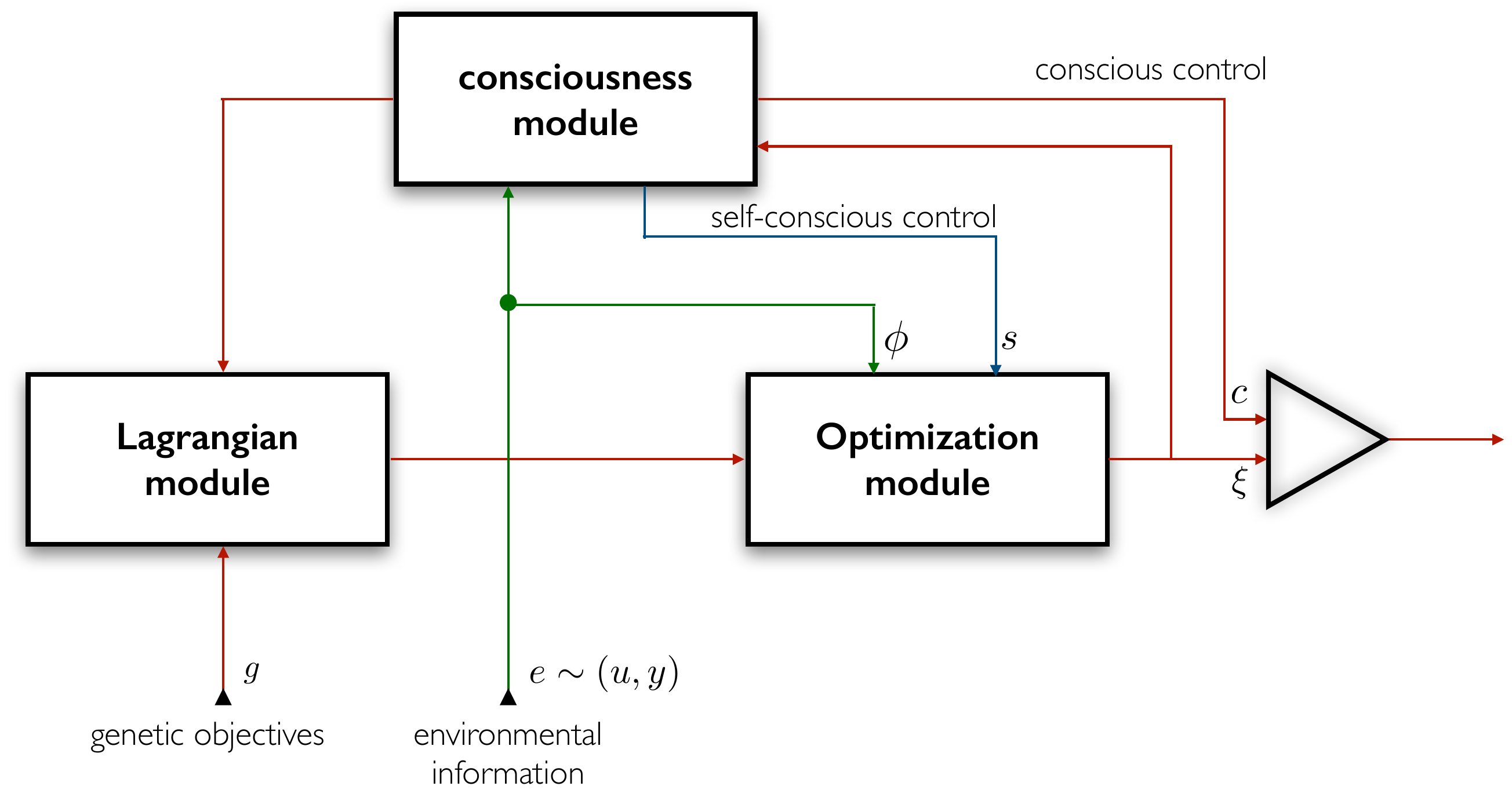}
\caption{\small {\em Overall architecture of {\em NARNIAN} agents.} The response of the {\em optimization module} is properly enabled by the {\em conscious module} when it returns $c=1$. On the opposite, when $c=0$ the response is disabled. In addition a similar gating mechanism is carried out on the Lagrangian term which somehow inhibit the learning process. The conscious model also carry out self-conscious control by returing function $s$.
}
\label{Conscious-Fig}
\end{figure}
\begin{remark} {\bf Lagrangian Gating Functions}\\
Like for the dynamical system which defines the agent, the Lagrangian function is characterized by the presence of gating functions that make it possible to properly select the corresponding dynamical behavior that arises from optimization. The triple $\zeta_{L} = \{m,\phi \}$ plays a dual role with respect to $\zeta_{N}$. As shown in the reminder of the paper, both gating functions play a crucial role in the overall process of learning. In general one can think of conscious mechanisms to control the learning process by using any of the gating parameters of $\zeta = \{\alpha, \psi, \omega; m,\phi \}$. However, in this paper we assume that conscious mechanisms influence learning by acting on $\phi$ only. This corresponds with restricting conscious issues to the process of enabling/disabling the effect of the loss function.    
\end{remark}

\begin{remark} {\bf Injecting invariance into the Lagrangian}\\
    The potential term ${\cal V}(\xi,e)$ in its general form returns a loss term
    depending on the degree of match between $\xi$ and $e$. Interestingly, there are important cases in which the dependency on the environment expressed by $e$ is dropped, so as we have that the potential gets the simplified structure ${\cal V}(\xi)$. Whenever this happens we are in front of ways to express invariant properties which involves the state only. A case in which this happens in computer vision corresponds with the expression of the brightness invariance. A more expressive non-holonomic constraint arises when modeling motion feature invariance. 
    In that case the features and their conjugated velocities satisfy a transportation pde which can be expressed by introducing a loss function which penalizes the mismatch of the constraint (see \cite{Betti_2022}, pp. 89--90).
    \iffalse
    Basically, this is a non-holonomic constraint that expresses knowledge on the
    environment. This can typically express an explicit knowledge on 
    the field $\tilde{\xi}_{\kappa}$ for $\kappa \in {\cal K}$. The transport
    equation in vision is an example of such a knowledge granule. 
    In that case those additional variables need to be consistent with 
    a correspondent set of neurons in $\bar{\cal V}$. If we denote
    by $\xi$ the field with units in $\bar{\cal V}$ then we can simply
    introduce in the Lagrange the correspondent loss  term  $1/2 (\xi-\tilde{\xi})^2)$.
     \fi
\end{remark}
As discussed in Appendix~\ref{appendix:control}, the optimization of $R(\nu,T)$ is well-posed when the boundary conditions are given and when the full access to the environmental information over $[0,\infty)$ is given. However, this is exactly what we exclude in the theory.  

In the following analysis, in addition to the risk function it turns out to be useful to introduce  the function
\textit{value function}  
\begin{equation}
    v(x,t) = \int_{t}^{T} L(x,\nu,s) ds, 
\label{vf_def}
\end{equation}
which is related to the risk by $R(\nu,T) = v(x(0),0)$.
Once the neural network~(\ref{eq:neural-model}) and the risk function~(\ref{eq:risk}) are given we can consider the following optimization problem
%\footnote{Notice that this is an optimization problem which needs to fulfil the non-holonomic %constraint~(\ref{eq:neural-model}).}
\begin{equation}
{\cal O}_{bc}: \ {\rm Boundary \ Conditions} \ \
\begin{cases}
    \nu_{[0,T]}^{\star} =  \argmin_{\nu} R(\nu,T);\\
    x(0) = x_{0};\\
    p(T) = v_{x}(x(T),T) = 0.
\end{cases}
\label{OptimalControl_T}
\end{equation}
Clearly, such an optimization problem can be solved by an oracle ${\cal O}_{bc}$
which, in addition to boundary conditions, has got the privilege of accessing all the environmental interactions of his life over $[0,T]$. 
An interesting notion of learning arises as we think of this optimization problem
in a causal framework. In this case we assume that $x(0) = x_{0}$ and
$p(0) = v_{x}(x(0),0) = 0$ are given and, in addition, we assume that
only environmental interactions up to a certain time $t$ are available, whereas
we cannot access the interactions over $[t,T]$.
This formulation is clearly ill-posed because of the lack 
of any assumption on future environmental interactions. It turns out that
those initial conditions do characterize the system dynamics which will 
end up with the required boundary condition $p(T) = v_{x}(x(T),T) = 0$ with
null probability. We notice in passing that the above mentioned discussing on gating parameters $\zeta$ suggests interpreting the risk function as 
${\cal R}(\nu,\zeta,T):=R(\nu,T)$. While we formulate learning as an optimal control problem driven by $\nu$, we are also allowed to modify the functional risk by an appropriate choice of the dissipation parameters $\zeta$.
Moreover, when formulating learning over $[0,\infty)$ we can think of attacking the following optimization problem:
\begin{equation}
{\cal O}_{fo}: {\rm Forward \ Optimization} \ \
\begin{cases}
    \nu^{\star}\in \argmin_{\nu} {\cal R}_{\infty}(\nu,\zeta);\\
    \zeta^{\star}\in \argmin_{\zeta} {\cal R}_{\infty}(\nu,\zeta);\\
    x(0) = x_{0}; \\ p(0) = p_{0};\\
    %\lim_{t \rightarrow \infty} p(t) = 0; \\ 
    \lim_{t \rightarrow \infty} \nu(t) = 0;\\
    %\lim_{t \rightarrow \infty} \alpha(t) = 0; \
    %\lim_{t \rightarrow \infty} \dot{\phi}(t) = 0; 
    %\lim_{t \rightarrow \infty} \dot{\zeta}(t) = 0. 
\end{cases}
\label{OptimalControlVel}
\end{equation}
This is a forward optimization problem based on Cauchy's initialization. Clearly, the formulation over $[0,\infty)$ cannot allow to use policies for solving the boundary problem and one can only expect to use the energy dissipation in such a way to guarantee that, as a consequence, $ \lim_{t \rightarrow \infty} p(t) = 0$. 
This opens the doors to  new mechanisms of learning according to which the neural propagation follows optimal control policies but also benefit from the appropriate adaptation of the $\zeta$ parameters. According to eq.~\ref{OptimalControlVel}, like the optimal control parameters $\nu$ those dissipation parameters must end up into convergence. During the agent's life, the $\zeta$ parameters are chosen in such a way to enforce the decrement of $1/2 p^2$ and/or the energy exchanged with the environment. 
\iffalse
An informal reading of the above optimization problem is that
the learning strategy consists of supplying the velocity of the weights under
the terminal conditions of discovering constant solutions for the 
weights $w_{ij}$, when the given functions $\alpha$, the conscious functions $\phi_{i}$, and the gating functions $\psi_{ij}$ are expected to converge asymptotically.
Those terminal condition corresponds with the underlying assumption that the 
learning environment offers a sort of regularities that can be captured 
by the dynamics of the recurrent neural network by constant weights. 
In particular, the terminal condition on $\nu$ removes asymptotically the kinetic energy and, 
consequently, which decreases the role of the masses $m_{ij}$.  
Moreover, it is worth mentioning that in this asymptotic case we loose
the uniqueness of the solution, that is $\nu^{\star}\in \argmin_{\nu} R(\nu,\infty)$ 
\fi

Finally we can   establish a  link with the classic functional risk in Machine Learning. 
Suppose we restrict the analysis to the case of unconscious learning for which
$\dot{\phi}=0$ holds true. Then the link arises when the formulated 
optimization problem~(\ref{OptimalControlVel}) admits the solution 
$\nu \rightarrow 0$ and $\dot{\zeta} \rightarrow 0$. 
In this case the risk~(\ref{eq:risk}) becomes
\[
    {\cal R}_{\infty}(\nu^{\star},\zeta^{\star}) = \int_{0}^{\infty} {\cal V}(\xi(t),u(t),y(t)) dt
    = \int_{\Xi \times Y} \mathrm{V}(u,y,{\cal N}) dP(u,y).
\]
Here, any environmental interaction $e(t) = (u(t),y(t))$ is transformed into 
the response $\xi(t),w(t)$ whose quality is measured by 
$\mathrm{V}(u,y,{\cal N})$, where  $\mathrm{V}$ plays the role of 
the loss function. The association  comes from considering all pairs $(u(t),y(t))$ 
generated over time and considering the associated joint probability. 

\section{Hamiltonian spatiotemporal locality}
\iffalse
Consider the following generalization of the  functional in \eqref{eq:risk}
\begin{equation}
\medmuskip-1mu
\begin{aligned}
R(\nu)=\int\limits_0^T \sum_{i \in \bar{\cal V}} \bigg[
        \sum_{j \in \bar{\cal V}} 
        \frac{m_{ij}}{2}\nu_{ij}^{2}(t) +
        \sum_{j \in \cal V} \phi_i(t) F\bigl(\xi_{i}(t;\nu),w_{ij}(t;\nu),t\bigr) \bigg] dt,
\end{aligned}
\end{equation}
where $\neuron(\cdot; \nu)$ and $\weight(\cdot;\nu)$ are the solutions
of Eq.~\eqref{eq:neural-model} once a set of $\nu_{ij}$ functions are
assigned.

In the reminder of the paper we will assume that 
\[
F(\neuron, \weight, t)=V(\neuron,t)+\lambda \Omega(\weight,t),
\]
where $V$ can be thought as a loss function on the output of the neurons
that may also enforce regularization directly on the neuron's outputs,
while $\Omega$ is a regularization term directly imposed on the weights of
the network (like for instance a weight-decay term).
\fi
We begin considering the optimal problem defined by Eq.~(\ref{OptimalControl_T}), for which we can offer the classic solution from the theory of Optimal Control (see Section~\ref{appendix:control} in the Appendix). 
The classic framework of the theory of optimal control suggests to express compactly the pair Neural Network - Lagrangian by the corresponding Hamiltonian function.  
It can be determined by considering the function
\begin{equation}
\begin{aligned}
&H(\neuron_{i},w_{ij},p_{i},p_{ij},t)\\
=&\min_v  \bigg(
        \frac{1}{2}\sum\nolimits_{ij} m_{ij}(t) \nu_{ij}^{2}(t)
        +  \phi_{i}(t)V\bigl(\xi_{i},t\bigr)\\
        &+\sum_i\alpha_i(t)\pneuron_i(t)
        \Big[-\neuron_i(t)+\sigma\Bigl(\sum\nolimits_j \omega_{ij}(t) w_{ij}(t)
        \neuron_j(t)\Bigr)\Big]\\
        &+\sum_{ij} p_{ij}(t)\psi_{ij}(t) \nu_{ij}(t)
        \bigg).
\end{aligned}
\end{equation}
Hence the minimum that defines the Hamiltonian is attained at
$\nu_{ij}= -(\psi_{ij}(t)/m_{ij}(t)) p_{ij}(t)$ which means that the Hamiltonian
can be computed in closed form
\begin{align*}
H(\xi_{i},w_{ij},p_{i},p_{ij},t)
    =&\frac{1}{2}\sum\nolimits_{ij} \frac{\psi_{ij}^2(t)}{m_{ij}(t)}
    p_{ij}^2(t) +  \sum_{i \in {\cal O}}\phi_{i}(t)V\bigl(\neuron(t),t\bigr) \\
    &+\sum\nolimits_i\alpha_i(t)\pneuron_i(t)
    \Big[-\neuron_i(t)+\sigma\Big(\sum\nolimits_j \omega_{ij}(t) w_{ij}(t)
    \neuron_j(t)\Big)\Big]
    -\sum\nolimits_{ij} \frac{\psi_{ij}^2(t)}{m_{ij}(t)} p_{ij}^2(t).\\
    =&-\sum\nolimits_{ij} \frac{\psi_{ij}^2(t)}{2m_{ij}(t)} p_{ij}^2(t) 
    +\sum\nolimits_i\alpha_i(t)\pneuron_i(t)
    \Big[-\neuron_i(t)+\sigma\Big(\sum\nolimits_j \omega_{ij}(t) w_{ij}(t)
    \neuron_j(t)\Big)\Big].
\end{align*}
Now if we define
\begin{equation}
    \beta_{ij}(t):=\frac{\psi_{ij}^2(t)}{m_{ij}(t)}
\end{equation}
the Hamiltonian becomes
\begin{align}
\begin{split}
H=&-\frac{1}{2} \sum\nolimits_{ij  \in {\cal A}}  \beta_{ij}(t)
    p_{ij}^2(t) +  \sum\nolimits_{i \in {\cal O}}\phi_{i}(t)
    V\bigl(\neuron(t),t\bigr)\\
    +&\sum\nolimits_{i \in \bar{\cal V}}\alpha_i(t)\pneuron_i(t)
    \Big[-\neuron_i(t)+\sigma\Big(a_{i}(t)\Big)\Big].
\end{split}
\label{H-eq}
\end{align}
\iffalse
\begin{remark}
    The Hamiltonian is a concise formalism for fully defining the system dynamics and, therefore, the neural propagation. 
    It also involves the dissipation parameters $\zeta$, while the control action is directly given by  $\nu_{ij}= -(\psi_{ij}(t)/m_{ij}(t)) p_{ij}$.
    It is worth mentioning that such an optimal control solution holds for any choice of the dissipation parameters, whose specific role has already been discussed. It turns out that for any choice of those parameters there is a corresponding optimal solutions. 
    \iffalse
    The optimal solution of oracle ${\cal O}_{2}$ only needs to 
    respect conditions~(\ref{OptimalControlVel}), that is those variables
    are given but must respect asymptotic convergence.
    On the opposite, in order to yield dissipation, 
    the causal solution involves all variables. 
    \fi
\end{remark}
\fi
\begin{remark} 
    A {\rm curiosity-driven process} can be carried out if the 
    intelligent agent is expected to process also derivatives of the
    input, that is if there are inputs for which 
    \begin{equation}
        \dot{\xi}_{i}(t) = u_{i}(t).
    \label{Curiosity-Driven-Eq}
    \end{equation}
    In this case the Hamiltonian needs to incorporate the corresponding 
    $p \cdot f$ terms, that is $\sum_{\cal I} p_{i}(t) u_{i}(t)$.  
\end{remark}

Since we have an explicit formula for $H$ we can directly compute the
gradients of the function. In particular we have that
\begin{lemma}
For all $t>0$ 
%and for all $(\neuron,\weight,\pneuron,\pweight)\in \R^D$
we have
\[\begin{aligned}
&\begin{aligned}
H_{\neuron_i}&=
 \phi_{i} V_{\neuron_i}(\neuron, t)-\alpha_i \pneuron_i
+\sum\nolimits_{\kappa \in {\rm ch}[i]}  \alpha_{\kappa} \omega_{\kappa i}
\sigma'\Bigl(\sum\nolimits_{j \in {\rm pa}[\kappa]} \omega_{\kappa j} w_{\kappa j}
\neuron_j\Bigr) w_{\kappa i} \pneuron_{\kappa};
\end{aligned}\\
& \begin{aligned}H_{w_{ij}}&=\alpha_i \omega_{ij}
\sigma'\Bigl(\sum\nolimits_{m \in {\rm pa}[i]} \omega_{im} w_{im}
\neuron_m\Bigr)\pneuron_i\xi_j;\end{aligned}\\
& H_{\pneuron_i}
=\alpha_i \Big[-\neuron_i+\sigma\Big(\sum\nolimits_j \omega_{ij} w_{ij}
\neuron_j \Big)\Big];\\
& H_{p_{ij}}=
- \beta_{ij} p_{ij}.
\end{aligned}
\]
\end{lemma}
\noindent This is all that is needed to write down the following Hamilton's equations of learning:

\begin{proposition}
For all $t\in(0,T)$ the following Hamilton's equations holds true
\begin{equation}
\left\{\begin{aligned}
&[i] \ \  \dot\neuron_i(t)=
\alpha_i(t) \Big[-\neuron_i(t)+\sigma\Big(\sum\nolimits_{j \in {\rm pa}[i]} 
\omega_{ij}(t) w_{ij}(t) \neuron_j(t)\Big)\Big];\\
&[ii] \ \  \dot{w}_{ij}(t)=-\beta_{ij}(t) p_{ij}(t);\\
&\begin{aligned}
[iii] \ \ \dot\pneuron_i(t)=& -s_{i}(t) \phi_{i}(t) V_{\neuron_i}(\neuron(t), t)+s_{i}(t) \alpha_i(t)
\pneuron_i(t)\\
 &-s_{i}(t) \sum\nolimits_{\kappa \in {\rm ch}[i]}
\alpha_{\kappa}(t) \omega_{\kappa i}(t)\sigma'\Bigl(\sum\nolimits_{j \in {\rm pa}[\kappa]} \omega_{\kappa j}(t)
w_{\kappa j}(t) \neuron_j(t)\Bigr)  w_{\kappa i}(t) \pneuron_{\kappa}(t);
\end{aligned}\\
&[iv] \ \ \dot{p}_{ij}(t)=
-s_{i}(t) \alpha_i(t) \omega_{ij}(t)\sigma'\Bigl(\sum\nolimits_m w_{im}(t)
\neuron_m(t)\Bigr)\pneuron_i(t)\xi_j(t).
\end{aligned}\right.
\label{rnn-costate-eq}
\end{equation}
\end{proposition}
Here we assume that $\dot{p}= s H_{x}$, where $s_{i}(t)=\pm 1$. This is one of the characteristic equation of the Hamilton-Jacobi-Bellman equations (see Section~\ref{appendix:control} in the Appendix) where $s_{i}(t)=1$. Interestingly, when flipping $s_{i}(t)$, the corresponding equation gives rise to an important system dynamical behavior based on dissipation which favor convergence. 
\iffalse
Furthermore we will assume in order to simplify notation that there is
no regularization terms on the weights ($\Omega\equiv0$).
With this simplifications Hamilton's equations become:
\begin{equation}
\left\{\begin{aligned}
&\dot\neuron_i(t)=
\alpha_i[-\neuron_i(t)+\sigma(a_i(t))];\\
&\dot\weight_{ij}(t)=-\frac{\psi_{ij}^2(t)}{m_{ij}(t)}\pweight_{ij}(t);\\
&\begin{aligned}
\dot\pneuron_i(t)=& -\phi_{i}(t) V_{\neuron_i}(\neuron(t), t)+\alpha_i
\pneuron_i(t)\\
&\quad{}-\sum\nolimits_k
\alpha_k\sigma'\bigl(a_k(t)\bigr) \pneuron_k(t)\weight_{ki}(t);
\end{aligned}\\
&\dot\pweight_{ij}(t)=
-\alpha_i\sigma'\bigl(a_i(t)\bigr)\pneuron_i\xi_j.
\end{aligned}\right.
\label{rnn-costate-eq}
\end{equation}
\fi
The optimal solution can be obtained by solving the above ODE under the boundary conditions:
\begin{equation}
    x^{\star}(0) = x_{0}, \ \ p^{\star}(T) = p_{T} = 0.
\end{equation}

%quiinizio
%%%%%%%%%%%%%%%%%%%%%%%%%%%%%%%%%%%%%%%%%%%%%%%%%%%%%%%%%%%%%%%%
Now we show that we can provide a simple straightforward description of the  evolution of the weights. We begin with the introduction of the another dissipation parameter $\theta_{ij}$ to facilitate such a description.  
\begin{definition}
    The real number 
    \begin{equation}
        \theta_{ij}:= \frac{\dot{m}_{ij}}{m_{ij}} 
        - 2 \frac{\dot{\psi}_{ij}}{\psi_{ij}}
    \end{equation}
    is referred to as a {\em joint dissipation factor}.
\label{ThetaDef}
\end{definition}
\noindent Here, ``joint'' means that, like for $\beta_{ij}$ also $\theta_{ij}$ involves dissipation that arises from both the model and the Lagrangian. 
Interestingly $\theta_{ij}$ and $\beta_{ij}$ are intimately related. 
\begin{lemma}
The dissipation parameter $\beta_{ij}$ evolves according to 
\begin{equation}
    \beta_{ij}(t) = \beta_{ij}(0)
    \cdot \exp{-\int_{0}^{t} \theta_{ij}(s) ds}
\end{equation}
\label{MonotinicityLemma}
\end{lemma}
\begin{proof}
The proof can be promptly driven
from the above the definition of $\beta_{ij}$. We have
\begin{equation}
    \dot{\beta}_{ij}
    = \frac{d}{dt} \frac{\psi_{ij}^{2}}{m_{ij}}
    = \frac{\psi^2_{ij}}{m_{ij}}
    \left(2\frac{\dot{\psi}_{ij}}{\psi_{ij}}
    - \frac{\dot{m}_{ij}}{m_{ij}}\right)= 
    -\frac{\psi^2_{ij}}{m_{ij}}\theta_{ij}
    =-\beta_{ij} \theta_{ij}.
\label{beta-theta-ODE}
\end{equation}
Then the thesis arises directly from the integration of the ODE.
\end{proof}
\noindent This leads to the following important result:
\begin{corollary}
    If $\theta_{ij}(t)>0$ then 
    $\dot{\beta}_{ij}(t) <0$ and $\beta_{ij}(t)>0$.
\label{beta_pos_dec}
\end{corollary}

\iffalse
\subsection{Convergence and generalization: $w_{ij}$ and $p_{ij}$}
Here we analyze the system dynamics in terms of the component of the state $w_{ij}$ and discuss the relationships with the corresponding co-state $p_{ij}$. We begin stating the driving result which leads to energy balance issues.
\fi
\begin{proposition} {\bf Second-order weights evolution}\\
    The Hamiltonian evolution dictated by Eqs.~\ref{rnn-costate-eq} leads to the following second-order direct interpretation of their dynamics 
    \begin{equation}
        \ddot{w}_{ij}+ \theta_{ij}\dot{w}_{ij}
        + \beta_{ij} \dot{p}_{ij} = 0.
    \label{eq:second-order}   
    \end{equation}
\label{Second_Order_w_prop}
\end{proposition}
\begin{proof}
The proof comes straightforwardly if we differentiate  with respect to time of the second of the Hamilton equations~\ref{rnn-costate-eq} when using Eq.~\ref{beta-theta-ODE} and Definition~\ref{ThetaDef}.
\end{proof}

%%%%%%%%%%%%%%%%%%%%%%%%%%%%%%%%%%%%%%%%%%%%%%%%%%%%%%%%%%%%%%%%%
%quifine

\noindent \emph{\sc Generalized Gradient} $\nabla_{\zeta}^{s}$\\
Eq.~\ref{rnn-costate-eq} -[ii] and eq.~\ref{eq:second-order} gives a straightforward expression for the evolution of the weights.  However, while Eq.~\ref{rnn-costate-eq} -[ii] connects the evolution of the weights to $p_{ij}$,  eq.~\ref{eq:second-order} yields a direct involvement of $\dot{p}_{ij}$. Interestingly, eq.~\ref{rnn-costate-eq} -[iii] and eq.~\ref{rnn-costate-eq} -[iv] can be thought of a {\em generalized gradients} based on loss $V$. In the first one, if we introduce $\nabla_{\zeta}^{s} V_{[0,t)}$ such that $\dot{w}_{ij} = - \beta_{ij} \nabla_{\zeta}^{s} V$ we can interpret $\nabla_{\zeta}^{s} V$ as the the driving gradient of the system dynamics. Moreover, this gradient-based interpretation leads to conclude that  
\begin{equation}
    \nabla_{\zeta}^{s} V_{[0,t)} := \int_{0}^{t} \big(s \odot \omega \odot p_{\xi} \odot \xi\big)(\tau) d\tau.
\end{equation}
Since, the system dynamics of eq.~\ref{rnn-costate-eq} -[iv] involves $\dot{p}_{ij}$ the corresponding gradient descent interpretation arises when the defining the generalized gradient by 
\begin{equation}
    \nabla_{\zeta}^{s} V (t) :=  \big(s \odot \omega \odot p_{\xi} \odot \xi\big)(t)
\end{equation}
\\
~\\
\noindent \emph{\sc Hamiltonian Learning and Spatiotemporal Locality}\\
The system dynamics defined by Eq.~(\ref{rnn-costate-eq}) exhibits \textit{locality both in time and space}, a property which has given rise to a longstanding discussion in the community of Machine Learning. The learning algorithms for recurrent neural networks are in fact mostly clustered under two different algorithmic frameworks that come from Backpropagation Through Time (BPTT) and Real Time Recurrent Learning (RTRL)
\cite{RumMcC86, williams1989learning}. Interestingly, BPTT is local in space, but not in time, whereas RTRL is local in time but not in space.
However, the appeal of spatiotemporal locality which is possessed by Eq.~\ref{rnn-costate-eq} will become relevant only when we will provide evidence that we can solve those ODE equations under Cauchy conditions according to the approximation stated by a causal agent.

\section{Co-state heuristics for learning $\zeta$}
Here we discuss some fundamental necessary conditions that we need to guarantee to allow an intelligent agent to learn. Basically, we need to upper bound both the state and co-state dynamics, which is a well-known problem in the optimal control theory. 

\subsection{Boundedness of the state $\xi_i$}
We begin stating a proposition on the BIBO stability of the recurrent neural network model described by ODE~\ref{eq:neural-model} which comes from a classic result of System Theory stated in the following Lemma.
\begin{lemma}
    Let is consider the ODE
    \begin{equation}
        \dot{x}(t) + \alpha(t) x(t) =  u(t).
    \label{LinODE}
    \end{equation}
    Then the solution can be expressed as
    \begin{equation}
        x(t) = x(0) \exp{\Big(-\int_{0}^{t} \alpha(\tau) d\tau}\Big)
        + \int_{0}^{t} \exp\Big(-\int_{s}^{t}\alpha(\tau) d\tau \Big)  u(s) ds 
    \end{equation}
\label{BoundednessLemma}
\end{lemma}
\begin{proof}
    Let us define the  integrating factor
    \[
        I(s):= \exp\Big(\int_{0}^{s} \alpha(\tau) d\tau\Big).
    \]
    If we multiply both sides of ODE~(\ref{LinODE}) we get
    \[
        I(s)\dot{x}(s) + \alpha(s) I(s) x(s) =  I(s)u(s).
    \]
    Now we have $D(I(s)x(s)) = x(s)\dot{I}(s)+I(s)\dot{x}(s)$ and then we get
    \[
        D(I(s)x(s)) -  x(s)\dot{I}(s)+ \alpha(s) I(s) x(s) =  I(s)u(s).
    \]
    Now we have $\dot{I}(s) = \alpha(s) I(s)$ and, therefore, 
    if we integrate over $[0,t]$ we have
    \[
        I(t)x(t) - I(0)x(0) = \int_{0}^{t} I(s)u(s) ds
    \]
    and, finally,
    \begin{align*}
        x(t) &= x(0)I^{-1}(t) + \int_{0}^{t} I^{-1}(t) I(s)u(s) ds\\
        &= x(0) \exp\Big(-\int_{0}^{t} \alpha(\tau) d\tau \Big)
        + \int_{0}^{t} I^{-1}(t) I(s)u(s) ds\\
        &= x(0) \exp\Big(-\int_{0}^{t} \alpha(\tau) d\tau \Big)
        + \int_{0}^{t} \exp\Big(-\int_{0}^{t}\alpha(\tau) d\tau 
        +\int_{0}^{s}\alpha(\tau) d\tau \Big) u(s) ds\\
        &= x(0) \exp\Big(-\int_{0}^{t} \alpha(\tau) d\tau \Big)
        + \int_{0}^{t} \exp\Big(-\int_{s}^{t}\alpha(\tau) d\tau  \Big) u(s) ds
    \end{align*}
\end{proof}
We can now promptly use this lemma for the BIBO stability of the recurrent neural network.
\begin{proposition} {\bf BIBO stability of the state} $\xi_i$\\
    Let us assume that 
    $\forall t \in [0,T], \forall i \in \bar{\cal V}: 
    \ \alpha_{i}(t) > \alpha > 0$. Then
    the recurrent neural network described by ODE~(\ref{eq:neural-model}) is BIBO stable.
\label{rnn_bibo_prop}
\end{proposition}
\begin{proof}
    Let $\mu_{i}(t)=\sigma(\sum_{j}\weight_{ij}(t) x_{j}(t))$ be. Then, from the boundedness of $\sigma$ we can always find $B_{\xi} \in \mathbb{R}^{+}$ such that $|\mu_{i}(t)| < B_\neuron$. As a consequence we have
\begin{align}
\begin{split}
        |\xi_{i}(t)| &= \bigg|\xi_{i}(0) \exp\Big(-\int_{0}^{t} \alpha_{i}(\tau) d\tau\Big) 
        + \int_{0}^{t} \exp\Big(-\int_{s}^{t}\alpha_{i}(\tau) d\tau  \Big) \mu_{i}(s) ds \bigg|\\
        &\leq 
        |\xi_{i}(0)| + \int_{0}^{t} e^{-\alpha(t-s)} |\mu_{i}(s)| ds \\
        &=|\xi_{i}(0)| + \frac{B_\neuron}{\alpha} (1-e^{-t}) < \xi_{i}(0) + \frac{B_{\xi}}{\alpha}.
\end{split}
\end{align}
\end{proof}
Notice that the strict inequality condition $\alpha(t)>\alpha>0$ is crucially used in the proof. For example, in the extreme case in which $\alpha(t) \equiv 0$ we have  
$|\xi_{i}(t)| = |\int_{0}^{t} \mu_{i}(s) ds|$ that is not necessarily bounded in the case of bounded input. Interestingly, the specific dynamical structure of the system dynamics of $\xi_{i}$ does allow to end up into the BIBO property also in case in which the related function $\alpha_i(t) \equiv 0$. In that case, from equation~\ref{eq:neural-model} we straightforwardly conclude that $\xi_{t} \equiv \xi_{i}(0)$.  This is in fact very important since the learning policies on the dissipation parameters only involves to force $\alpha_i(t) = 0$. As already mentioned, this nicely matches neurobiological evidence on the presence of inhibited neurons. 

\subsection{Co-state boundedness and $\zeta$ heuristics}
In addition to bounding the state, clearly the system dynamics imposes also to bound the co-state. This also corresponds with the need of bounding the energy exchanged with the environment that, as it will be shown in the remainder of the paper, also involves the expression of the Hamiltonian. \\ 
~\\
\emph{\sc Analysis on} $p_{i}$\\
We are interested in analysing dynamical modes for which $d/dt \big((1/2) p_{i}^{2}\big)<0$. From eq.~\ref{rnn-costate-eq} we get 
\begin{align*}
\pneuron_i(t) \dot\pneuron_i(t)=& -s_{i}(t) \phi_{i}(t) V_{\neuron_i}(\neuron(t), t) \pneuron_i(t) +s_{i}(t) \alpha_i(t)
\pneuron_i^{2}(t)\\
&-s_{i}(t) p_{i}(t) \sum\nolimits_{\kappa \in {\rm ch}[i]}
\alpha_{\kappa}(t) \omega_{\kappa i}(t)\sigma'\Bigl(\sum\nolimits_{j \in {\rm pa}[\kappa]} \omega_{\kappa j}(t)
w_{\kappa j}(t) \neuron_j(t)\Bigr)  w_{\kappa i}(t) \pneuron_{\kappa}(t);
\end{align*}
We can promptly end up into the following proposition:
\begin{theorem}
    The co-state is upper-bounded $p_{i}(t)$ by choosing $\alpha_{i}(t) = 0$ and $\phi_{i}(t)=0$.
\label{p-boundedness-prop}
\end{theorem}
This is quite an obvious result since the policy conditions correspond with disabling the system dynamics and the potential of the Lagrangian term. We notice in passing that the disabling dynamical condition of the system dynamics $\alpha_{i}(t) \equiv 0$ could be replaced with $m_{i}(t) \equiv 0$ to continue to guarantee the boundedness of the co-state. This is in fact a straightforward consequence of the nullification of the value function~\ref{vf_def}, which leads to an ill-posed problem. Interestingly, such a conclusion does not arise from Proposition~\ref{p-boundedness-prop} which is inherently considering system dynamics that is consistent with the given optimization problem. 

Now let us analyze conditions to guarantee that $d/dt \big((1/2) p_{i}^{2}\big)<0.$
This can be guaranteed by acting on the $\zeta$ dissipation parameters.
Let us distinguish the case $i \in {\cal O}$ with respect to $i \in {\cal H}$.
\begin{itemize}
\item $i \in {\cal O}$: 
In this case the above decrement term can satisfied by enforcing: 
\begin{equation*}
\omega_{ii}(t) w_{ij}(t)
< \frac{s_{i}(t) +  \phi_{i}(t) V_{\neuron_i}(\neuron(t), t) \pneuron_i(t)}{\alpha_i(t) \pneuron_i^{2}(t) \big(1-
\sigma'\big(a_{i}(t)\big)} 
\end{equation*}
When considering the underlying assumption that $i$ is enabled $\alpha_{i}(t)>0$ we can promptly see that this condition can be satisfied by any appropriate choice of the dissipation parameter $\omega_{ii}$. Of course, the satisfaction of the condition can also jointly obtained when considering the other dissipation functions $\alpha_i$ and $\phi_i$. 

\item $i \in {\cal H}$: In this case the decrement of the co-state requires the satisfaction of 
\[
\sum\nolimits_{\kappa \in {\rm ch}[i]}
    \big[\alpha_{\kappa}(t) \sigma'\big(a_{k}(t)\big)  w_{\kappa i}(t)  p_{i}(t)  \pneuron_{\kappa}(t)\big] \cdot \omega_{\kappa i}(t) > s_{i}(t)
\]
Like for the case $i \in {\cal O}$ this condition can always be met by an appropriate choice of the dissipation parameter $\omega_{ij}$.
\end{itemize}
Now let us define $\forall t \in [0,+\infty): \ {\rm p}_{i}(\zeta,t):=p_{i}(t)$, which makes the dependence on $\zeta$ explicit. Moreover, let us consider the square of the overall unit co-state $p_{\xi}^{2}:= \sum_{i} p_{i}^{2}$.
Hence, we can summarize this analysis by stating the following theorem:
\begin{theorem}
    For any $t \in [0,\infty)$, there always exists a choice of the $\zeta(t)$ dissipation parameters such that
    \[
    \frac{d}{dt} \frac{1}{2}{\rm p}_{\xi}^{2}(\zeta,t) <0
    \]
\label{px-decrement-th}
\end{theorem}
The space of possible choices of $\zeta$ previously discussed to satisfy the above condition clearly increase when considering the whole space of parameters and the joint condition on $p_{\xi}$.

%%%%%%%%%%%%%%%%%%%%%%%%%%%%%%%%%%%%%%%%%%%%%%%%%%%%%%%%%%%%%%%%%%%%%%%%%%%%%%%%
% $p_{ij}$ analysis
%%%%%%%%%%%%%%%%%%%%%%%%%%%%%%%%%%%%%%%%%%%%%%%%%%%%%%%%%%%%%%%%%%%%%%%%%%%%%%%%
~\\
\emph{\sc Analysis on} $p_{ij}$\\
Let us consider the evolution of $p_{ij}$ expressed by
\[
    \dot{p}_{ij}(t) = - s_{i}(t) \alpha_{i}(t) \omega_{ij}(t) \sigma^{\prime}(a_{i}(t)) p_{i}(t) \xi_{j}(t).
\]
We can carry out the dual analysis shown for $p_{i}$ by considering the evolution of $(1/2) p_{ij}^{2}$. We have
\[
    p_{ij}(t)\dot{p}_{ij}(t) = - s_{i}(t) \alpha_{i}(t) \omega_{ij}(t) \sigma^{\prime}(a_{i}(t)) p_{ij}(t) p_{i}(t) \xi_{j}(t)
\]
and, therefore, 
\[
    \frac{d}{dt} \frac{1}{2}p_{ij}^{2}(t) = -  \alpha_{i}(t) \sigma^{\prime}(a_{i}(t)) [ p_{ij}(t)] \cdot [ s_{i}(t) \omega_{ij}(t)  p_{i}(t) \xi_{j}(t)].
\]
We can enforce the decrement by imposing 
\[
    p_{ij}(t) \cdot s_{i}(t)  \omega_{ij}(t) p_{i}(t) \xi_{j}(t) >0.
\]
Like in the case of state $xi$, if we relax the satisfaction for all arcs and simply consider $p_{w}^{2} = \sum_{ij} p_{ij}^{2}$ then we end up in the following condition: 
\begin{equation}
    \left\langle p_{w}, s \odot \omega  \odot p_{\xi} \odot \xi  \right \rangle(t) >0.
\end{equation}

We can summarize the previous analysis in the following theorem.
\begin{theorem}
    The decrement of the co-state kinetic term obeys the following condition: 
    \[
    \frac{d}{dt} \frac{1}{2}p_{w}^{2}(t) < 0 \leftrightarrow 
    \left\langle  p_{w}, \nabla_{\zeta}^{s} V  \right \rangle(t) >0
    \]
\label{pw-decrement-th}
\end{theorem}
This theorem states that the decrement of  the co-state kinetic term $(1/2)p_{w}^{2}$ takes place if and only if the  co-state $p_{w}$ and the generalized gradient $\nabla_{\zeta}^{s} V$ are correlated.

\noindent Now we state another important proposition which helps understanding the progress of learning. 
\begin{proposition}
The velocities of the state $w$ and co-state $p_{w}$ are correlated if and only if the square of the co-state decreases, that is
\begin{equation}
    \langle \dot{w},\dot{p} \rangle = \dot{w} \cdot \dot{p}  > 0 
    \leftrightarrow \frac{d}{dt} p_{w}^{2} < 0.
\end{equation}
\end{proposition}
\begin{proof}
    The proof is straightforward.
    From Hamiltonian equation $\dot{w}_{ij} = - \beta_{ij} p_{ij}$ we get 
    \[
    \dot{w}_{ij} \dot{p}_{ij} = - \beta_{ij} p_{ij} \dot{p}_{ij} = - \beta_{ij} 
    \frac{1}{2} \frac{d}{dt} p_{ij}^{2}.
    \]
    Hence, the proof follow when considering that
    $\beta_{ij}>0$.
\end{proof}

%\subsection{$\zeta$ learning }
This analysis shows that we can always decrease the co-state by appropriate choices of $\zeta$. 
\iffalse
This suggests adopting the classic gradient descent scheme 
\begin{equation}
    \dot{\zeta} = - \eta \nabla_{\zeta} \frac{1}{2} {\rm p}^{2}(\zeta,t),
\end{equation}
where ${\rm p}^{2} = {\rm p}_{\xi}^{2} + {\rm p}_{w}^{2}$, where $\neta>0$ is the learning rate. When considering the corresponding policies on ${\rm p}_{\xi}^{2}$  and ${\rm p}_{w}^{2}$ and the results stated by Theorem~\ref{px-decrement-th} and Theorem~\ref{pw-decrement-th} we get 
\begin{align*}
    \frac{d}{dt} \frac{1}{2}p_{w}^{2}(t) = 
    \nabla_{\zeta} \frac{1}{2} {\rm p_{w}}^{2}(\zeta,t) \cdot \dot{\zeta} = 
     - \eta \Big(\nabla_{\zeta} \frac{1}{2} {\rm p_{w}}^{2}(\zeta,t)\Big)^2
    < 0 \leftrightarrow 
    \left\langle  p_{w}, \nabla_{\zeta}^{s} V  \right \rangle(t) >0
\end{align*}
\fi
When $d/dt ({\rm p}_{ij}^2) \rightarrow 0$ then $\dot{p}_{ij} \rightarrow 0$.
From Proposition~\ref{Second_Order_w_prop} we get $\ddot{w}_{ij} + \theta_{ij} \dot{\theta}_{ij} = 0$ which, in turns, leads to achieve constant solutions for $w_{ij}$. Now, if $\beta_{ij} \neq 0$, for the Hamiltonian learning~\ref{rnn-costate-eq}-$iv$ we get $p_{ij} \rightarrow 0$.

\iffalse
Now let us introduce the following differential operator
\begin{equation}
    {\cal D}(w_{ij},\theta_{ij}):= \ddot{w}_{ij} + \theta_{ij} \dot{w}_{ij}.
\end{equation}
From Proposition~\ref{Second_Order_w_prop} We have 
\[
    p_{ij} {\cal D}(w_{ij},\theta_{ij}) + \frac{1}{2} \beta_{ij} \frac{d}{dt} p_{ij}^{2} = 0.
\]
Let $\beta_{ij}(0)>0.$ We can promptly see that the condition $\frac{d}{dt} p_{ij}^{2}<0$ yields
\[
    p_{ij} {\cal D}(w_{ij},\theta_{ij}) >0
\]
\fi

\iffalse
\section{Energy-driven causal agents}
In this section we start thinking of the causal agent. We carry out an analysis which driven by
the consequences of bounding the Hamiltonian. Since it depends on the state and on the 
costate, we discuss their boundedness separately. \\
\fi
%%%%%%%%%%%%%%%%%%%%%%%%%%%%%%%%%%%%%%%%%%%%%%%%%%%%%%%%%%%%%%%%%%%
%qui 16 agosto: theorema su s=\pm 1 stesso comportamento se p(0)=0. 
% Uso Eulero e dimostro per induzione
%%%%%%%%%%%%%%%%%%%%%%%%%%%%%%%%%%%%%%%%%%%%%%%%%%%%%%%%%%%%%%%%%%%
\section{Energy balance}
\label{eb_sec}
Now we carry out a classic analysis on the energy balance coming from the interaction
of the agent with the environment. We begin considering the contribution $H_{t}$.
In case $s_{i}(t) \equiv 1$ we have $H_{x} \cdot \dot{x} + H_{p} \cdot \dot{p} = 0$, which leads to $\dot{H} = (d/dt) H = H_{t}$.
\begin{definition}
    The terms
    \begin{align}
    \begin{split}
    E &:=  \int_{0}^{t} \phi_{i}(\tau) V_{s}(\xi_{i}(\tau),\tau)  d\tau\\
    D&:=D_{\phi}+D_{\beta}+D_{\alpha}+D_{\omega}
    \end{split}
    \end{align}
    are referred to as the {\em environmental energy} and the {\em dissipated energy}, respectively, where 
    \begin{align}
    \begin{split}
    D_{\phi}&:=-  \int_{0}^{t}\dot{\phi}_{i}(\tau) V(\xi_{i}(\tau),\tau) d\tau\\
    D_{\beta}&:=\frac{1}{2}\int_{0}^{t} \sum\nolimits_{ij} \dot{\beta}_{ij}(\tau) p_{ij}^2(\tau) d\tau\\
    D_{\alpha}&:=- \int_{0}^{t}\sum\nolimits_{i \in \bar{\cal V}}\dot{\alpha}_i(\tau)\pneuron_i(\tau)
    \Big[-\neuron_i(\tau)+\sigma\big(a_{i}(\tau)\big)\Big] d\tau\\
    D_{\omega}&:=-\int_{0}^{t}\sum\nolimits_{i \in \bar{\cal V}}\alpha_i(\tau) p_i(\tau) \sigma^{\prime}\big(a_{i}(\tau)\big)
    \sum\nolimits_{j}\dot{\omega}_{ij}(\tau) w_{ij}(\tau)  \neuron_j(\tau) d\tau
    \end{split}
    \end{align}
\label{EnergyTerm_def}
\end{definition}
\noindent The dissipation energy arises because of the temporal changes of 
$\alpha_{i}, \psi_{ij}, m_{i}, \phi_{i}$, even though the role of $\psi_{ij}$ is replaced
with $\beta_{ij}$. 
\begin{theorem} - {\bf I Principle of Cognidynamics}\\
    The system dynamics evolves under the energy balance
    \begin{equation}
        E = \Delta H + D
    \end{equation}
\end{theorem}
\begin{proof} 
\begin{align*}
    H_{s}\big|_{s=\tau}&=
    \frac{\partial}{\partial s} \Big(
    -\frac{1}{2}\sum\nolimits_{ij} \beta_{ij}(s)
    p_{ij}^2(s) +  \sum_{i \in {\cal O}}\phi_{i}(s) V\bigl(\neuron(s),s\bigr)\\
    %+\lambda  \sum_{ij} \phi_{i}(s) w_{ij}^{2}(s) \\
    &+\sum\nolimits_{i \in \bar{\cal V}}\alpha_i(s) p_i(s)
    \Big[-\neuron_i(s)+\sigma\Big(\sum\nolimits_j \omega_{ij}(s) w_{ij}(s)
    \neuron_j(s)\Big)\Big] \Big)\Big|_{s=\tau}\\
    &=-\sum\nolimits_{ij} \dot{\beta}_{ij}(\tau) p_{ij}^2(\tau)
    +  \sum_{i \in {\cal O}} \phi_{i}(\tau) V_{s}(\xi_{i}(\tau),\tau)
    +\dot{\phi}_{i}(\tau) V(\xi_{i}(\tau),\tau)\\
    &%+\lambda  \sum_{ij} \dot{\phi}_{i}(\tau)w_{ij}^{2}(\tau)
    +\sum\nolimits_{i \in \bar{\cal V}}\dot{\alpha}_i(\tau) p_i(\tau)
    \Big[-\neuron_i(\tau)+\sigma\Big(\sum\nolimits_j \omega_{ij}(\tau) w_{ij}(\tau)
    \neuron_j(\tau)\Big)\Big] \\
    &+\sum\nolimits_{i \in \bar{\cal V}}\alpha_i(\tau) p_i(\tau) \sigma^{\prime}\big(a_{i}(\tau)\big)
    \sum\nolimits_{j}\dot{\omega}_{ij}(\tau) w_{ij}(\tau)  \neuron_j(\tau)  
\end{align*}
Now, let $\Delta H:=H(\xi(t),w(t),p_{\xi}(t),p_{w}(t),t)
-H(\xi(0),w(0),p_{\xi}(0),p_{w}(0),0)$ be. 
If we integrate over $[0,t]$ we get
\begin{align*}
    \Delta H &= \int_{0}^{t} 
    \sum_{i \in {\cal O}} 
    \big(\phi_{i}(\tau) V_{s}(\xi_{i}(\tau),\tau) \big) d\tau
    & \leftarrow E\\
    &+ \int_{0}^{t} \dot{\phi}_{i}(\tau) V(\xi_{i}(\tau),\tau) d\tau
    &\leftarrow -D_{\phi}\\
    &-\frac{1}{2}\int_{0}^{t} \sum\nolimits_{ij} \dot{\beta}_{ij}(\tau) p_{ij}^2(\tau) d\tau
    &\leftarrow -D_{\beta}\\
    &%+ \int_{0}^{t} \lambda \gamma \sum_{ij} 
    %\dot{\phi}_{i}(\tau)\weight_{ij}^{2}(\tau)  
    +\int_{0}^{t}\sum\nolimits_{i \in \bar{\cal V}}\dot{\alpha}_i(\tau)\pneuron_i(\tau)
    \Big[-\neuron_i(\tau)+\sigma\big(a_{i}(\tau)\big)\Big] d\tau
    &\leftarrow -D_{\alpha}\\
    &+\int_{0}^{t}\sum\nolimits_{i \in \bar{\cal V}}\alpha_i(\tau) p_i(\tau) \sigma^{\prime}\big(a_{i}(\tau)\big)
    \sum\nolimits_{j}\dot{\omega}_{ij}(\tau) w_{ij}(\tau)  \neuron_j(\tau) d\tau
    &\leftarrow -D_{\omega}
\end{align*}
\end{proof}
\noindent This theorem states a formal result which follows the spirit of the energy conservation principle. The environmental energy $E$ turns out to be exchanged with the agent in terms of the variation of the internal energy $\Delta H$ and dissipated energy\footnote{ 
We notice in passing that if we flip the signs of all the energy terms of Definition~\ref{EnergyTerm_def} the I Principle of Cognidynamics still holds true.}D.
We can promptly see that any process of learning corresponds with a decrement of $E$. The I principle of Cognidynamics suggests that such a decrement results in energy dissipation that is characterized by a negative term of $D$. For example, from eqs.~\ref{EnergyTerm_def}, we can promptly see that the increment of $\phi$ and the decrement of $\beta$ yields the required energy dissipation. Hence, the I Principle clearly explains why we need to properly change the dissipative weights for an effective process of learning. 

\section{Gravitational neural networks}
A fundamental problem that plagues neural network-based approaches to lifelong learning is that when attacking new tasks they typically offer no guarantees against catastrophically adapting learned weights that were already used for successfully solving previously learned tasks. While the proposed Hamiltonian learning scheme offers a truly new scheme of learning for recurrent neural networks, in principle, it shares this shortcoming with related gradient-based methods. The problem seems to have an architectural origin that certainly remains in dynamic neural networks. Furthermore, it is worth mentioning that the natural dynamic behavior suitable for the interpretation of cognitive processes is to continuously generate trajectories in the phase space. In other words, while some neurons can be deactivated, convergence to fixed points is not biologically plausible.
It is very interesting to note that these needs are satisfied simply by an appropriate dynamic structure of the system which requires sustaining of trajectories through the presence of conservative processes.
\begin{figure}[H]
	\centering
	\includegraphics[width=11cm]{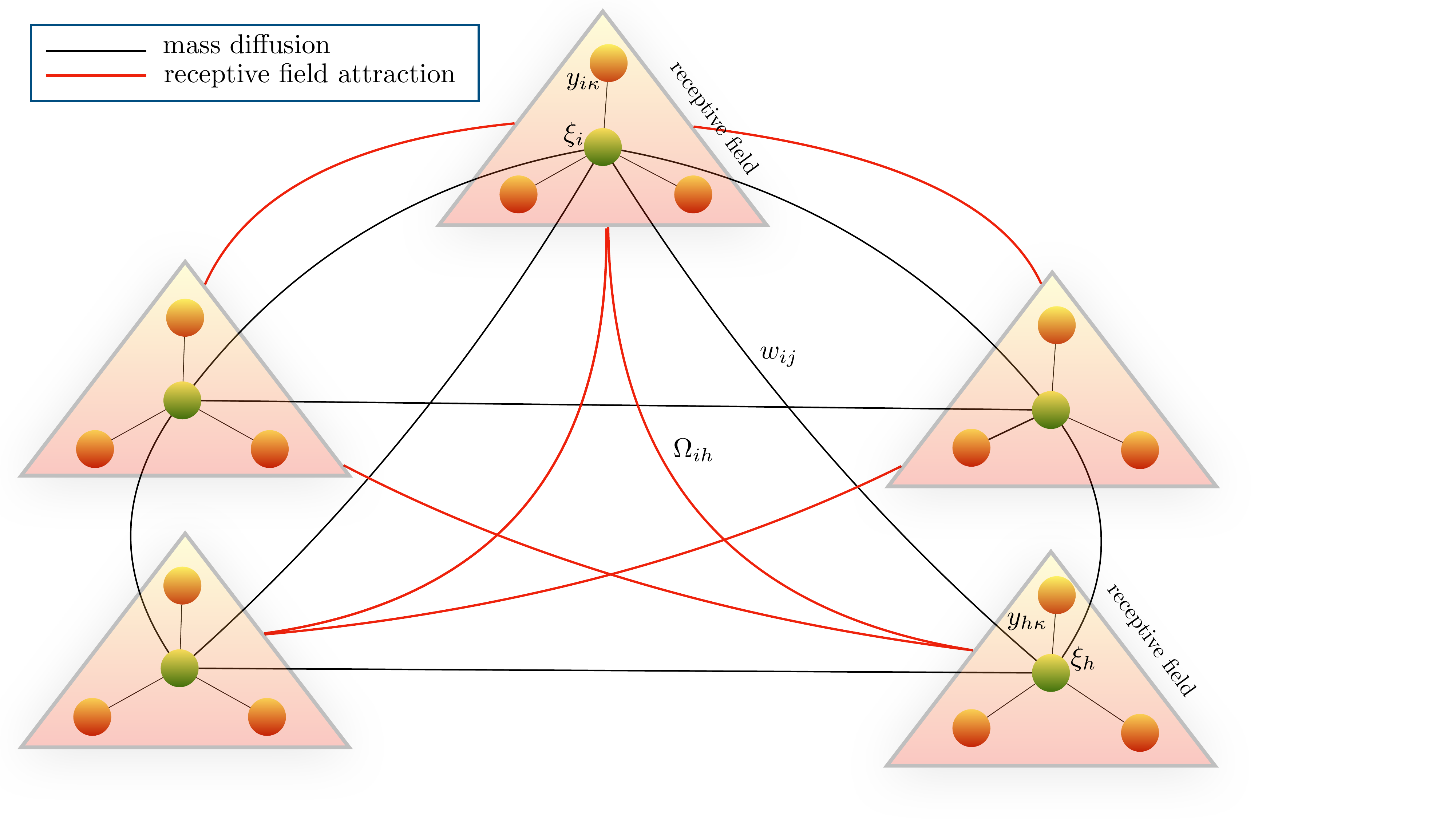}
\caption{\small {\em The dynamics of neural propagation is driven by the interaction of receptive fields composed of neurons whose outputs represent their coordinates. The interaction is inspired by classical Newtonian gravitational dynamics which is spatially localized since it vanishes at large distances. Such a localization property leads to generate dynamical structures that under appropriate classic conditions are stable, thus providing the support for a sort of addressable memory, which attacks by design issues of catastrophic forgetting.}
}
\label{GNN-Fig}
\end{figure}
Interestingly, this can be carried out by embedding $\xi_{i}$ into a trajectory which somehow represents its memory address. In doing so the neurons considered so far might be regarded as bodies that are characterized by their own trajectory that arises from classic gravitational interaction. Hence, one can think of the body mass $\xi_{i}$ along with its own trajectory $y_{i \kappa}$, where $\kappa=1,\ldots, d$ is the dimension of the space when the n-body interaction takes place. This suggests considering a $d$-dimensional receptive fields that, like bodies, are characterized by classic Newtonian dynamics

\begin{equation}
\left\{\begin{aligned}
&\dot{y}_{i \kappa}(t) = {\rm y}_{i\kappa}(t);\\
& \dot{\rm y}_{i\kappa}(t)= - (d-2) G \sum_{h \in {\rm ne}[i]} 
\underbrace{
\frac{y_{i\kappa}-y_{h\kappa}}
{
\Big[
\sum_{\kappa=1}^{d} (y_{h\kappa}-y_{i\kappa})^{2}
\Big]^{d/2}
}}_{\Omega_{ih}} \sigma(\xi_{h});\\
&\dot{Y}_{i\kappa}(\omega,t) = e^{-i\omega t} y_{i\kappa}(t)\\
\end{aligned}\right.
\label{rnn-cotate-eq}
\end{equation}

Here the state becomes $(\xi_{i},w_{ij},y_{i\kappa})$. The memory allocation is compactly expressed by $[y_{i\kappa}] \leftarrow \xi_{i}$.
Finally, a frequency-based response is well-suited to report decision-based information. 

\section{Conclusions}
This paper focuses on the interpretation of natural learning in the optimization framework of dynamic programming which gives rise to Hamiltonian-Jacobi-Bellman equations.  The paper promotes a collectionless approach to Machine Learning that strongly parallels what happens in nature and uses basic ideas that are massively adopted in Theoretical Physics. The system dynamics which drives the learning process is in fact dictated by Hamiltonian ODE, which turns out to parallel 
classic gradient descent methods in Statistical Machine Learning.

The most important contribution of the paper is that of addressing the longstanding questions on the stability of learning in recurrent neural networks, and to show that the answer comes from the introduction of an appropriate law which drives the control of dissipative weights.  
The proposed mathematical framework offers the appropriate tools for understanding the exchange of energy between the agent and the environment and suggests that the process of dissipation decreases the entropy of the system, thus creating ordered structures. In particularly, it becomes clear that we need an appropriate developmental scheme which requires filtering the inputs coming from the environment for offering a well-posed formulation of learning.
\iffalse
Most importantly, the energy balance clearly suggests the introduction of focus attention mechanisms with a twofold role: First, they are useful to bound the error of the agent by checking the co-state and by performing time-reversal processing at some critical points.  Second, the neural connections are paired with a gating mechanism that is driven by the energy balance to approximate the optimal solution that would require the knowledge of boundary conditions. 
\fi

Interestingly, the adoption of the proposal Hamiltonian learning approach leads to also to a computational scheme that, unlike BPTT and RTRL, is local in both time and space.

The proposed theory makes use of the continuous setting of computation to interpret learning as the discovery of a stationary point of the cognitive action, which tightly parallels the interpretation of Newtonian laws in Mechanics.  
This facilitates the development of the main results of the paper and offers the substrate for investigating links with Developmental Psychology and Neuroscience.
However, we are mostly planning to work towards  the translation in the discrete setting of computation of the proposed learning approach, which is in fact  quite straightforward. It can open the doors to any application of Machine Learning involving time, where the emphasis is moved to the collectionless approach joined with the central role of focus of attention. 

\section{Acknowledgements}
I mostly thank Stefano Melacci, Alessandro Betti, Michele Casoni, and Tommaso Guidi for insightful discussions arising during lab meetings at SAILab. 
Early ideas can be traced back to the discussions I had with Tomaso Poggio  concerning his seminal paper ``Regularization Networks'' after a seminar I gave at MIT on 2011. More significant steps were carried out after the NeurIPS 2020 Workshop and the LOD 2021 Workshop on Biological Plausibility in Neural Computation, especially from the interaction with Alessandro Sperduti, Yoshua Bengio, Naftali Tishby, and Tomaso Poggio. 
Finally, I thank Jay McClelland for a recent discussion we had during the Third Conference on Lifelong Learning Agents (CoLLAs 2024)  on early studies on dynamical systems for connectionist models. 

%\bibliography{corr,nn}
%\bibliographystyle{plain}
%\bibliographystyle{icml2024}

%%%%%%%%%%%%%%%%%%%%%%%%%%%%%%%%%%%%%%%%%%%%%%%%%%%%%%%%%%%%%%%%%%%%%%%%%%%%%%%
%%%%%%%%%%%%%%%%%%%%%%%%%%%%%%%%%%%%%%%%%%%%%%%%%%%%%%%%%%%%%%%%%%%%%%%%%%%%%%%
% APPENDIX
%%%%%%%%%%%%%%%%%%%%%%%%%%%%%%%%%%%%%%%%%%%%%%%%%%%%%%%%%%%%%%%%%%%%%%%%%%%%%%%
%%%%%%%%%%%%%%%%%%%%%%%%%%%%%%%%%%%%%%%%%%%%%%%%%%%%%%%%%%%%%%%%%%%%%%%%%%%%%%%
\newpage
\appendix
 
\newpage
\appendix
\begin{center} \large
{\bf APPENDIX}
\end{center}
\section{HJB equations}
\label{appendix:control}
Let us consider\footnote{In this paper we overload the notation by using 
a symbol like $x$ to denote a variable as well as the corresponding
function which return it at time $t$.} of
the {\em Value Function} $V: [0,T] \times {\cal X} \to \bbR: \ (t,\xi) \mapsto V(t,x)$
\begin{equation}
    V(t,x):= J_{T} + \min_{w}\int_{t}^{T} ds \
    L(\xi(s),w(s),s).
\label{L-def}
\end{equation}
 Here, $\xi(s)$ is the trajectory in $[t,T]$ driven by 
 \begin{equation}
 	\dot{\xi}(s) = f(\xi(s),w(s),s)
\label{BasicODE}
\end{equation}
which begins with $\xi(t)=x$.
Function $V$ is sometimes also referred to as the {\em cost-to-go}.
Here $J_{T} \geq 0$ is the final value that might be regarded as 
\[
	J_{T} = \int_{T}^{\infty} ds \ L(\xi(s),w(s),s).
\]
Basically, the introduction of $J_{T}$ leads to consider the special case in which
$T=\infty$ when there exists the integral 
$V(t,x):=  \min_{w}\int_{t}^{\infty} ds \   L(\xi(s),w(s),s)$.
We begin with a couple of premises on $(f,L)$ that are very important
in the following.
\begin{itemize}
\item 
%Depending on the pair  $(f,L)$ the value function might be defined. 
The mentioned case of Mechanics is a classic example in
which the action, in general, does not admit minimum. 
%In those cases
%one can replace $\min_{w}$ with $\delta_{w}$.
\item Function $f$ and $L$ are supposed to be continuous and differentiable with
	respect to $x,w$ whereas we make no assumption on the continuity 
	with respect to $t$.
\end{itemize}
The optimum is determined by using Bellman's principle. 
We consider the general case in which $f$ and $L$ posses 
analytical regularities only with respect to $x$, that is we assume 
that $f(x,w,t)$ and $L(x,w,t)$ admit continuous partial derivates 
with respect to $x$ only. In particular we assume that $w$ and $t$
only posses a finite number of discontinuities and that 
$f(x,w,t), L(x,w,t)$ are always bounded in $[0,T]$. 
Under this assumption we can always grid $[0,T]$ in such a 
way that the mentioned discontinuities correspond with nodes
in the grid. 

Given $\Delta t>0$, we want to see the relationship between the 
value function at $t$ and at $t+\Delta t$, where the optimal value on 
the trajectory $x^{\star}$ is correspondently 
moved to $x^{\star} + \Delta x^{\star}$.
We have
%We have\footnote{As usual, 
%the notation ``$dt, dx$'' neglects higher-oder infinitesimals.}
\begin{align*}
    V(t,x^{\star})\hspace{-0.1cm}&=\hspace{-0.1cm}
    \min_{w([t,T])}\hspace{-0.1cm} \bigg(V(t+\Delta t,x+\Delta x)+
    \int_{t}^{t+\Delta t} ds \ L(x(s),w(s),s)
    \bigg)
\end{align*}
and look for the control policy $w^{\star}([0,T])$ over $[0,T]$.
If we apply Bellman's principle we get 
\begin{align*}
	V(t,x^{\star}) &\hspace{-0.1cm}
	= V(t+\Delta t,x^{\star}+\Delta x^{\star}) + 
	\min_{w([t,t+\Delta t])} L(x(t),w(t),t) \Delta t + o(\Delta t)\\
	&=V(t,x^{\star})+ V_{s}(t,x^{\star}) \Delta t+
	V_{x}(t,x^{\star}) \Delta x^{\star} 
	+ o(\Delta x^{\star}) + o(\Delta t)\\
	&+
	\min_{w([t,t+\Delta t])} L(x(t),w(t),t) \Delta t,
\end{align*}
where $x^{\star}=x(t)$.
This holds true for points $t+\Delta t$ where we can use the local 
expansion of $V(t+\Delta t,x^{\star}+\Delta x^{\star})$. Interestingly, 
this can always be done also in case of the mentioned discontinuities
when setting the jumps of $w$ and the jumps of the Lagrangian 
because of the explicit dependence on $t$ exactly at $t+\Delta t$.
Under these assumptions, on the optimal trajectory we have 
$
	\Delta x^{\star} = f(x^{\star},w^{\star},t) \Delta t + o(\Delta t),
$
where $w^{\star}=\arg \min_{w} V(t,x^{\star})$.
%	= \min_{w_{[t,t+\Delta t]}}  f(x,w,t) dt 
%\]
Hence, the previous equation can be re-written as 
\begin{align*}
	o(\Delta t) &=V_{x}(t,x^{\star}) \hspace{-0.5mm} \cdot \hspace{-0.5mm} 
	f(x^{\star},w^{\star},t) \Delta t  + V_{s}(t,x^{\star}(t)) \Delta t +
	\min_{w([t,t+\Delta t])} L(x(t),w(t),t) \Delta t
\end{align*}
Now, as $\Delta t \rightarrow 0$ we have  $\min_{w([t,t+\Delta t])}  \leadsto \min_{w(t)}$.  Let $\omega:=w(t)$. Then  we get
\begin{equation}
    V_{s}(t,x^{\star}) = 
    - \min_{\omega}
    \bigg(
        L(x^{\star},\omega,t)  
        + V_{x}(t,x^{\star}) \cdot f(x^{\star},\omega,t)
    \bigg).
\label{BellmanEq}
\end{equation}
From this analysis, we are now ready to state the following theorem.
\begin{theorem}
Suppose we are given the optimization problem on the value function
defined by eq.~(\ref{L-def})
with the terminal boundary 
condition $\forall x \in \bbR^{n}: \ \ V(T,x) = g(x)$. 
If we define\footnote{In general we need to replace 
$\min$ with {\rm stat}.} 
\begin{equation}
	H(x,p,s):= \min_{\omega}\big( 
	 L(x,\omega,s) + p \cdot f(x,\omega,s)
	\big),
\label{H-def}
\end{equation}
then any optimal trajectory satisfies the {\em Hamilton-Jacobi-Bellman (HJB)} equations
\begin{equation}
	V_{s}(t,x^{\star})+ H(x^{\star},V_{x}(t,x^{\star}),t) =0.\\
\label{HJB-ODE-Eq}
\end{equation}
\boxed{}
\label{HJB-Theorem}
\end{theorem}
The previous discussion on the adoption of the Bellman's principle 
holds for any $x \in \bbR^{n}$ at time $T$.
When considering that $H$ is defined by Eq.~(\ref{H-def}), the above
equation is a partial differential equation 
\begin{equation}
	V_{s}(t,x) + H(x,V_{x},t)=0
\label{H-def-HJB}
\end{equation}
on $V(t,x)$ which can be solved under the respect of the 
terminal condition $V(T,x)=g(x)$. 
The solution of this PDE returns $V(t,x)$ that can be determined along with 
the optimal policy 
$
	w^{\star} = \arg \min_{\omega} \big( 
	 L(x,\omega,s) + p \cdot f(x,\omega,s)
	\big)
$.
It is wort mentioning that the partial differential equation~(\ref{H-def-HJB})
which leads to the discovery of $V(t,x)$ is in fact characterized by the
Hamiltonian function which dictates the corresponding evolution 
of $V(t,x)$.\\
 ~\\
{\sc Algorithmic scheme for HJB eqs}
\begin{enumerate}
\item		Set $\partial \mathscr{X}=\{x \in \bbR^{n}: \ \ V(T,x)=g(x)\}$
\item		Solve HJB eqs~(\ref{H-def-HJB}) 
		$V_{s}(t,x) + H(x,V_{x},t)=0$
		with boundary condition on $\partial \mathscr{X}$ at $t=T$
		and, thereby compute $V$;
\item 	Select $w^{\star}(t)$ which returns the minimum of 
		${\cal H}(x,p,w,t)$, that is 
		\[
			w^{\star}(t) = \arg \min_{w(t)} {\cal H}(x(t),p(t),w(t),t)
			 = \arg \min \big(
			 	L(x,w,t) + p^{\prime}(t) \cdot f(x,w,t)
			 \big)
		\]
\end{enumerate}
%%%%%%%%%%%%%%%%
% independence of bias terms
%%%%%%%%%%%%%%%%
Now we put forward properties that arises from the 
bias independence property of shifting the Lagrangian with a constant
term, that is $L \leadsto \tilde{L}:=L+c$.
\begin{proposition}
	If we shift the Lagrangian, that is $L \leadsto \tilde{L}:=L+c$ then:
	\begin{itemize}
	\item [i.] The optimal solution $w^{\star}$ is the same and 
		$\tilde{V}(t,x)=c(T-t)+V(t,x)$;
	\item [ii.]The same shift holds for the Hamiltonian, that is 
		$\tilde{H}(x,p,t) = H(x,p,t) + c$.
	\end{itemize}
\label{H-shift-invariance}
\end{proposition}
\begin{proof}
$[i]$: The optimal solution is clearly independent of the shift with term $c$.
	Moreover,  then the value function becomes
	\begin{align*}
		\tilde{V}(t,x) &=\tilde{J}_{T} + \int_{t}^{T} ds \ \big(c + L(x(s),w(s),s) \big) = 
		c(T-t) + \int_{t}^{T} ds \ \big(L(x(s),w(s),s) \big)\\
		&=c(T-t) + J_{T} \int_{t}^{T} ds \ \big(L(x(s),w(s),s) = 
		c(T-t) + V(t,x).
	\end{align*}
$[ii]$. Concerning the Hamiltonian we have
\begin{align*}
	\tilde{H}(x,p,t) &= \tilde{L}(x,w,t) + \tilde{V}_{x}(t,x) \cdot f(x,w,t)  \\
	&= c + L(x,w,t) + V_{x}(t,x) \cdot f(x,w,t) = c + H(x,p,t).
\end{align*}
\boxed{}
\end{proof}
\begin{remark}
Finally, we can promptly see the effect of adding a bias to the Hamiltonian
in the HJB equations. Hence, suppose that $H \leadsto \tilde{H}=H+c$.
We have 
$\tilde{V}_{s}(t,x) +c+ H(x,p(t),t) = 0$ and
$V_{s}(t,x) + H(x,p(t),t) = 0$, from which we get 
\[
	\tilde{V}_{s} - V_{s} = \frac{\partial}{\partial s}(\tilde{V}-V) = -c.
\]
Hence we can determine the value function corresponding from 
Hamiltonian shifting by
$\tilde{V}(t,x)-V(t,x) = C - ct$, where $C \in \bbR$ 
can be determined when imposing that 
$\tilde{V}(T,x) = V(T,x)$, that is $C=cT$. Finally we get
$\tilde{V}(t,x) = V(t,x) + c(T-t)$ which is consistent with the 
first statement of Proposition~\ref{H-shift-invariance}.
\end{remark}

%%%%%%%%%%%%%
% feedback property
%%%%%%%%%%%%%
\noindent \emph{\sc control policy: feedback of the state}\\
Now, we show that optimal policies $w^{\star}$ can be determined on the basis
of the knowledge of the pair $(t,x^{\star})$ only. Basically,
there exist $\alpha$ such that 
\begin{equation}
	w^{\star}(t) = \alpha(x^{\star}(t),t)
\end{equation}
For any pair $(x^{\star}(t),t))$, function $\alpha(\cdot,\cdot)$ is
defined by 
\[
		\alpha(x^{\star}(t),t) = \arg \min_{\omega} 
		\big( L(x,\omega,s) + p \cdot f(x,\omega,s)\big).	
\]
Hence, we can characterize optimal strategies by the following theorem
which introduces the general {\em feedback control strategy}.
\begin{theorem}
	Let us consider the terminal optimization problem with the boundary
	condition $\forall x \in \bbR^{n}: \ V(T,x)=g(x)$. Then the
	optimal strategy is a {\em feedback control} that is characterized by
	\begin{align}
		\begin{split}
			&\alpha: ({\cal X} \subset \bbR^{n}) 
			\times [0,T] \to {\cal W} \subset \bbR^{m}:
			(x^{\star}(t),t) \mapsto \alpha(x^{\star}(t),t).\\
			&\dot{x}^{\star}(t) = f(x^{\star}(t),\alpha(x^{\star}(t),t),t)
		\end{split}
	\end{align}
\label{FC-th}
\end{theorem}
\begin{remark}
	The theorem makes it possible to express the optimal value function
	by
	\[
		V(t,x^{\star}(t)) = \int_{t}^{T} ds \
		L(x^{\star}(s),\alpha(x^{\star}(s),s),s).
	\]
	This suggests that if we are interested in the asymptotic value 
	we need to the corresponding asymptotic condition on the
	Lagrangian
	\[
		\lim_{s \to \infty} L(x^{\star}(s),\alpha(x^{\star}(s),s),s)=0.
	\]
\end{remark}

%
% H =0 and time independence
%
Now suppose that $f$ and $L$ are {\em time invariant}, that
is we are in the case $f: {\cal X} \times {\cal W} \to \bbR_{0}^{+}: 
\ (x,w) \mapsto f(x,w)$ and
$L: {\cal X} \times {\cal W} \to \bbR_{0}^{+}: \ (x,w) \mapsto L(x,w)$. Under this condition Theorem~\ref{FC-th} leads to establish
the following corollary.
\begin{corollary}
	if the pair $(f,L)$ is time invariant $H$ is also time invariant and, 
	moreover, we have $H(x,p)\equiv 0$ or, equivalently, $H$ is constant.
\label{H=0-corollary}
\end{corollary}
\begin{proof}
	We begin noticing that $H$ inherits  time-invariance from $(L,f)$.
	Under the hypothesis the control law established by Theorem~\ref{FC-th}
	leads to the expression $w^{\star}(t) = \alpha(x^{\star}(t)$. Hence,
	the corresponding value function is
	\[
		V(t,x^{\star}) = \int_{t}^{T} ds \ L(x^{\star}(s),\alpha(x^{\star}(s)))
		= \hat{V}(x^{\star}).
	\]
	From Theorem~\ref{HJB-Theorem} (HJB equations) we get
	\[
		V_{s}(t,x^{\star}) + H(x^{\star},V_{x}(t,x^{\star}),t) = 0
		\rightarrow H(x^{\star},V_{x}(t,x^{\star})=0.
	\]
	Finally, the thesis comes from $p(t)=\hat{V}_{x}(x^{\star})$.\\
\boxed{}
\end{proof}

\begin{remark}
The spirit of HJB equations is that of determining the learning policy 
in the framework of a problem with terminal boundary conditions, 
thus reflecting the basic ideas of dynamic programming. As already
pointed out, the HJB equations are PDE which can be solved 
under terminal boundary conditions. In many real-world problems
their direct usage can be prohibitive because of the associated
computational complexity. However, the case covered in 
Corollary~\ref{H=0-corollary} offers a dramatic simplification since
we need to look for trajectories such that $H(x(t),p(t))=0$ 
({\em Hamiltonian invariance}). We can 
promptly see that 
\begin{align}
\begin{split}
	\dot{x}(t) &= H_{p}(x(t),p(t))\\
	\dot{p(t)} & -H_{x}(x(t),p(t))
\end{split}
\end{align}
satisfy the Hamiltonian invariance expressed by the 
Poisson's brackets $H(x(t),p(t))= \left\{H,H \right\}=0$. We have that 
the null total derivative
\begin{align*}
	\frac{d}{dt} H(x(t),p(t))= H_{x}(x(t),p(t)) \cdot \dot{x}(t) 
	+ H_{p}(x(t),p(t)) \cdot \dot{p}(t) =0
\end{align*}
holds for all $t \in (0,T)$, which leads to conclude that $H(x(t),p(t)) \equiv 0$
in the interval. In the following we want to see if this fundamental
property holds also in the case in which the Hamiltonian is 
time-dependent. We will address this issue by 
using the {\em method of characteristics}.
\end{remark}
\iffalse
Interestingly, its solution comes with  the optimal trajectory
$x^{\star}(t), \ t \in [0,T]$ that corresponds with $w^{\star}$ determined by Eq.~(\ref{H-def}).
This is the fundamental step behind HJB equations, since it allows us to 
determine the solution by collapsing the partial differential equation 
to an ordinary differential equation when pairing Eq.~(\ref{HJB-ODE-Eq})
with Eq.~(\ref{H-def}). This dramatic simplification also suggests to 
simplify the notation, so as, in the following,  Eq.~(\ref{H-def})
and Eq.~(\ref{HJB-ODE-Eq}) will be replaced with
\begin{equation}
\begin{cases}
	H= L + V_{x} \cdot f\\
	V_{t} + H=0.
\end{cases}
\label{HJB-equations}
\end{equation}
When using this notation we are implicitly assuming that all functions are computed
on the optimal trajectory.
\fi
\section{Method of characteristics} 
The HJB approach to optimization assumes that one knows the boundary
conditions at the end-point of the interval. Unfortunately, in that form,
they are neither useful for conception nor for the understanding of
learning schemes 
Now we will shown that classic Hamiltonian dynamics that satisfies
the HJB equations for time-independent Hamiltonians also works
for the general case of time-variant Hamiltonians.\\
 ~\\
 \noindent{\sc Hamiltonian dynamics is sufficient}\\
Let us consider the following (HJ)  initial-point problem
\begin{equation}
({\rm HJ}) \ \ \ \ \ 
\begin{cases}
	V_{s}(t,x)+ H(x,V_{x}(t,x,t)) =0.\\
	V(0,x) = g(x).
\end{cases}
\label{JB-ODE-Eq}
\end{equation}
We want to convert this PDE problem into an ODE that can open 
a dramatically different computational perspective. We use the
method of characteristic.
Now, let us introduce the {\em co-state} $p$ as 
$
	p:= V_{x}
$ and consider the total derivative\footnote{We use Einstein notation.} of its  $\kappa$ coordinate
\begin{align}
	\dot{p}_{\kappa}(t)=
	%:=\dot{p}_{x_{\kappa}}(t) = 
	V_{x_{\kappa} t}(t,x(t)) +
	V_{x_{\kappa}x_{i}} \cdot \dot{x}_{i}. 
\label{p-c-eq}
\end{align}
Now, if $V$ solves (HJ) then
\begin{align*}
	V_{x_{\kappa} t}(x,t) = - H_{x_{\kappa}}
	(x,V_{x}(x,t),t) - 
	H_{p_{i}}(x,V_{x}(x,t),t) \cdot
	V_{x_{i} x_{\kappa}}(x,t).
\end{align*}
Now if we plug $V_{x_{\kappa} t}(x,t)$ 
in Eq.~(\ref{p-c-eq}) we get
\begin{align}
\begin{split}
	\dot{p}^{\kappa}(t)&=
	 - H_{x_{\kappa}}(x(t),\underbrace{V_{x}(x(t),t)}_{p(t)},t) \\
	 &+ \big(\dot{x}_{i}(t) -
	H_{p_{i}}(x(t),\underbrace{V_{x}(x(t),t)}_{p(t)},t)\big) \cdot
	V_{x_{\kappa}x_{i}}(t,x(t)). 
\end{split}
\label{no-simpl}
\end{align}
Now we can promptly see that the following choice
\begin{equation}
({\rm H}) \ \ \ \ \
\begin{cases}
	\dot{x}(t) = H_{p}(x(t),p(t),t)\\
	\dot{p}(t) = -H_{x}(x(t),p(t),t)
\end{cases}
\label{H-equations}
\end{equation}
satisfies Eq.~(\ref{no-simpl}). %and, therefore, HJB equations.

Now we prove that if we can solve ({\rm H}) there is in fact 
an appropriate initialization which leads to solve also ({\rm BH})
(Eq.~\ref{JB-ODE-Eq}). In a sense, it is sufficient that ({\rm H})
holds true to solve ({\rm HJ}).
Suppose we initialize with $p^{0}=g_{x}(x^{0})$ and solve ({\rm H})
with $x(0)=x^{0}$ and $p^{0}=g_{x}(x^{0})$.  
\iffalse
Any trajectory $x(t)$ that satisfies HJB is such that
\[
	\forall t \in (0,T): \
	\frac{d}{dt} \bigg(
		 \underbrace{V_{t}(t,x(t)) + H(t,x^{\star}(t),p^{\star}(t))}_{\tt HJB(t)}
	\bigg)= 0.
\]
When plugging the Hamiltonian dynamics we get
\begin{align*}
	\frac{d}{dt} {\tt HJB}(t)&= 
	\frac{\partial}{\partial t} \frac{d}{dt}
	\bigg(
		 V_{t}(t,x(t)) + H(t,x^{\star}(t),p^{\star}(t))  
	\bigg) = \frac{\partial}{\partial t}
	\bigg(
		- L(x(t),w(t),t)\bigg) \\
	& + \frac{\partial}{\partial t}  \bigg(
	H_{t}(x^{\star}(t),p^{\star}(t),t) 
	+ H_{x}(x^{\star}(t),p^{\star}(t),t) \cdot \dot{x}^{\star}
	+ H_{p}(x^{\star}(t),p^{\star}(t),t) \cdot \dot{p}^{\star}
	\bigg)\\
	&= - L_{t}(x(t),w(t),t) + \frac{\partial}{\partial t} (L(x(t),w(t),t)+p(t) \cdot
	f(x(t),w(t),t))
\end{align*}
\fi
Now let us express
\begin{align}
\begin{split}
	\frac{d}{dt} V(x(t),t) &= V_{t}(t,x(t)) + V_{x}(t,x(t)) \cdot \dot{x}(t) \\
	&=- H(x(t),p(t),t) + p(t) \cdot H_{p}(x(t),p(t),t).
\end{split}
\label{crucial-step}
\end{align}
We also have $V(0,x(0))=V(0),x^{0}=g(x^{0})$ and, therefore, if we 
integrate over $t$ we get\footnote{This is in fact an ODE that allows us
to determine the value function once we have determined the solution of
optimization problem via an integral that depends on the Hamiltonian.}
\begin{align}
	V(t,x(t)) = g(x^{0}) + \int_{0}^{t} ds (-H(x(s),p(s),s) + 
	H_{p}(x(s),p(s),s) \cdot p(s)). 
\label{dVdt-exp}
\end{align}
%Notice that the derivation of {\rm H} equations relies on the crucial step 
%following from Eq.~(\ref{crucial-step}).
% Are there other solutions
%different from {\rm H} which lead to {\rm HJ}? In the following
%we prove that {\rm H} is necessarily coming from {\rm HJ}.
From this analysis we can relate HJB equations to the ODE
{\em Hamiltonian} equations as stated in the following theorem.
\begin{theorem}
	Let us consider the minimization problem defined by~(\ref{L-def})
	with initial conditions V(0,x)=g(x).
	Then, like for Theorem~\ref{HJB-Theorem}, 
	 the optimal solution $w^{\star}$ is determined by 
	eq.~\ref{H-def}. Moreover, the two conjugated variables
	 $x,p$  that satisfy ODE~(\ref{H-equations}) are also solutions
	 of the initial {\rm (HJ)} problem~(\ref{JB-ODE-Eq}). 
\label{H-theorem}
\end{theorem} 
%
% Boundary conditions on x(0), p(T)
%
The analysis that arises from the method of characteristics establishes
a deep connection between the solution of the (HJB) equations and the 
(H) Hamiltonian equations regardless of the given formulation as a 
Cauchy problem. In particular, if we  know the value of $p(T)$
then we can establish the optimality.
\begin{theorem}
	Let us consider the minimization problem defined by~(\ref{L-def})
	with boundary conditions $x(0)=x_{0}$ and $p(T)=p_{T}$.
	Then the solution of ODE~(\ref{H-equations}) is also the
	solution {\rm HJB} problem~(\ref{HJB-ODE-Eq}), that is
	the solution of the minimization problem~(\ref{L-def}).
\label{H-th-BoundaryConditions}
\end{theorem}

% studiare il rapporto tra la condizione terminale su p e il valore J_T.

\section{Hamilton equations and Lagrangian multipliers} 
A possible way to attack optimization 
under  constraints is to use the Lagrangian 
approach and find the stationary points of 
\begin{equation}
	J_{L} = J_{T} + \int_{0}^{T} dt \bigg(
		L\big(x(t),w(t),t\big) + 
		\lambda(t) \cdot \big(f(x(t),w(t),t))-\dot{x}(t)\big)
		\bigg).
\end{equation}
We introduce the Hamiltonian on the optimal trajectory by setting
\[
	{\cal H}\big(x(t),\lambda(t),w(t),t\big):=L\big(x(t),w(t),t\big) + \lambda(t) \cdot f(x(t),w(t),t),
\]
in such a way to re-write $J_{L}$ as 
\begin{equation}
J_{L}(x,\lambda) = J_{T} + \int_{0}^{T} dt \bigg(
		\underbrace{{\cal H}(x(t),\lambda(t),w(t),t) - \lambda(t) \cdot \dot{x}(t)}_{{\cal L}^{x}}
		\bigg).
\label{J-H-exp-x}
\end{equation}
Now, in order to determine a stationary solution of $J_{L}$,
if we use by part integration, we can promptly see that 
we can replace $\lambda(t) \cdot \dot{x}(t)$ with $- \dot{\lambda}(t) \cdot x(t)$.
We have 
\begin{align*}
	\int_{0}^{T}dt \ \lambda(t) \cdot \dot{x}(t) = \bigg[\lambda(t) \cdot x(t)\bigg]_{0}^{T}
	- \int_{0}^{T} dt \ \dot{\lambda}(t) \cdot x(t)
\end{align*}
Hence we reformulate the problem of determining the stationary solution of 
\begin{equation}
	J_{L}(x,\lambda) = J_{T} - \bigg[x(t) \cdot \lambda(t)\bigg]_{0}^{T}+\int_{0}^{T} dt \big(
		\underbrace{{\cal H}\big(x(t),\lambda(t),w(t),t\big) + x(t) \cdot \dot{\lambda}(t)}_{{\cal L}^{\lambda}}
		\big).
\label{J-H-exp-lambda}
\end{equation}
When using the Euler-Lagrange equations on functional $J_{L}$ it is convenient to 
use both its representations given by eq.~(\ref{J-H-exp-x})  and eq.~(\ref{J-H-exp-lambda}).
We get
\begin{align*}
	&0=\frac{d}{dt} {\cal L}_{\dot{x}}^{x} - {\cal L}_{x}^{x} 
	\rightarrow \ \ \dot{\lambda}(t) + {\cal H}_{x}(x(t),\lambda(t),w(t),t) = 0\\
	&0=\frac{d}{dt} {\cal L}_{\dot{\lambda}}^{\lambda} - {\cal L}_{\lambda}^{\lambda} 
	\rightarrow \ \ \dot{x}(t) - {\cal H}_{\lambda}(x(t),\lambda(t),w(t),t) = 0\\
	&0=\frac{d}{dt} {\cal L}_{\dot{w}}^{\lambda} - {\cal L}_{w}^{\lambda} 
	\rightarrow \ \ {\cal H}_{x}(x(t),\lambda(t),w(t),t) = 0.
\end{align*}
Clearly, the last equation can also be found by using ${\cal L}_{x}^{x}$.
The third condition leads to marginalize $w$. In particular, when the stationary point
corresponds with a minimum we have 
\begin{equation}
	H(x,\lambda,t) = \min_{w} {\cal H}(x,\lambda,w,t).
\end{equation}
Finally, this leads to the Hamiltonian equations
\begin{equation}
\begin{cases}
	\dot{\lambda}(t) =- H_{x}(x(t),\lambda(t),t) \\
	\dot{x}(t) = - {\cal H}_{\lambda}(x(t),\lambda(t),t). 
\end{cases}
\end{equation}

Basically, the Lagrangian multiplier $\lambda$ is formally equivalent to the co-state 
$p$. Finally, we need to see the role of the boundary conditions that arise when
we use the Euler-Lagrange equations. 
When involving ${\cal L}^{x}$  the boundary conditions are 
\begin{align}
\begin{split}
	x(0)=x_{0} \\
	{\cal L}_{x}^{x} = \lambda(T)=\lambda_{T}.
\end{split} 
\label{UsefulBC}
\end{align}
Likewise, for symmetry, when involving ${\cal L}^{\lambda}$ we get
\begin{align}
\begin{split}
	\lambda(0) = \lambda_{0}\\
	x(T) = x_{T}.
\end{split}
\end{align}
In most problems we make use of conditions~(\ref{UsefulBC}).

\section{Links with Analytic Mechanics}
Let us consider the following causal optimization
\begin{align}
\begin{split}
    \dot{x} &= v \\
    L(x,v) &= \frac{1}{2} m v^2 + V(x)
\end{split}
\label{MechEB}
\end{align}
where $v$ is regarded as the control variable.
Then we can determine the Hamiltonian by defining
${\cal H}(v) = \frac{1}{2} m v^2 + V(x) + p v$. Now, the 
minimum is achieved for $v= - p/m$ and, therefore, we have
$H = V(x) - p^{2}/2m$.
Let us see the outcome of causal optimization $s=-1$. We have
\begin{align}
\begin{split}
    \dot{x} &= H_{p} = - \frac{p}{m}\\
    \dot{p} &= H_{x} = V_{x},
\end{split}
\end{align}
from which we get the Newtonian equations
\[
    \ddot{x} = - \frac{\dot{p}}{m} = - \frac{V_{x}}{m}
    \rightarrow m \ddot{x} + V_{x} = 0.
\]
Now we want to see how the temporal evolution of the Hamiltonian.
If we consider the classic Lagrangian $L(x,v)=1/2 \cdot m v^2 - V(x)$
then $H(x,p) = p^{2}/2m + V(x)$ and we know that 
\[
    \frac{dH}{dt} = [H,H] = 0.
\]
When offering the interpretation of causal optimization we
have 
\begin{align}
    \frac{d}{dt} H(x,p,t) &= 
    H_{x} \dot{x} + H_{p} \dot{p} + H_{t}
    = 2 \dot{p} \cdot \dot{x} + H_{t}
\end{align}
and, therefore, when considering
$H = - \frac{1}{2m} p^{2} + V(x)$, since $H_t = 0$,  we have
$dH/dt = 2 \dot{x} \dot{p}$. Clearly, the Hamiltonian is not 
interpreted as the energy anymore. However, also in this case, 
we promptly end up into the 
classic energy conservation principle
\[
    \frac{d}{dt}H(x(t),p(t))= 2 \dot{x}(t) \dot{p}(t) =  
    -2  \frac{p(t)}{m} \dot{p}(t) 
    =  - \frac{1}{m} \frac{d}{dt} p^{2}(t).
\]
that is 
\[
    \frac{d}{dt}
    \Big(H(x(t),p(t)) 
    + \frac{1}{m}p^{2}(t) \Big) =
    \frac{d}{dt} \Big(
        \frac{1}{2m} p^{2}(t) + V(x(t))
    \Big)
    =0.
\]

\noindent  \emph{\sc Dissipation}\\
Let us consider the following causal optimization
\begin{align}
\begin{split}
    \dot{x} &= v \\
    L(x,v) &= \Big(\frac{1}{2} m v^2 + V(x)\Big) e^{\theta t}
\end{split}
\label{MechEB2}
\end{align}
where $v$ is regarded as the control variable.
Then we can determine the Hamiltonian by defining
${\cal H}(v) = (\frac{1}{2} m v^2 + V(x))e^{\theta t} + p v$. Now, the 
minimum is achieved for $v= - (p/m) e^{-\theta t}$ and, therefore, we have
$H = V(x) e^{\theta t} - e^{-\theta t} p^{2}/2m$.
Let us see the outcome of causal optimization $s=-1$. We have
\begin{align}
\begin{split}
    \dot{x} &= H_{p} = - \frac{p}{m} e^{-\theta t}\\
    \dot{p} &= H_{x} = V_{x} e^{\theta t},
\end{split}
\end{align}
from which we get the Newtonian equations
\[
    \ddot{x} = - \frac{\dot{p}}{m} e^{-\theta t} 
    + \theta \frac{p}{m} e^{-\theta t} = - \frac{V_{x}}{m} - \theta \dot{x}
    \rightarrow m \ddot{x} + \theta \dot{x} + V_{x} = 0.
\]
Now we want to see how the temporal evolution of the Hamiltonian.
We have 
\begin{align*}
    \frac{d}{dt} H&= \Big(
        V e^{\theta t} - \frac{1}{2m}e^{-\theta t} p^{2}
    \Big)\\
    & e^{\theta t} \frac{d}{dt}V + \theta V e^{\theta t}
    +\frac{1}{2m} \theta e^{-\theta t} p^{2}
    - e^{-\theta t} \frac{d}{dt} \Big(\frac{p^{2}}{2m} \Big)\\
    &= e^{\theta t} \Big(
    \frac{d}{dt}V + \theta V 
    +\frac{1}{2m} \theta e^{-2\theta t} p^{2}
    - e^{-2\theta t} \frac{d}{dt} \Big(\frac{p^{2}}{2m} \Big)
    \Big)
\end{align*}

When offering the interpretation of causal optimization we
have 
\begin{align}
    \frac{d}{dt} H(x,p,t) &= 
    H_{x} \dot{x} + H_{p} \dot{p} + H_{t}
    = 2 \dot{p} \cdot \dot{x} + H_{t}
\end{align}
and, therefore, when considering
$H = - (1/2m) e^{-\theta t} p^{2} + e^{\theta t} V(x)$,  we have
\begin{align*}
    dH/dt &= \frac{1}{2m} \theta  e^{-\theta t} p^{2} 
    + \theta e^{\theta t} V + 2 \dot{x} \dot{p}\\
    &=\frac{1}{2m} \theta  e^{-\theta t} p^{2} 
    + \theta e^{\theta t} V - 2 \dot{p} \frac{p}{m} e^{-\theta t}
\end{align*}
Clearly, the Hamiltonian is not 
interpreted as the energy anymore, but when replacing the expression of $H$
we get
\begin{align*}
  &e^{\theta t} \frac{d}{dt}V + \theta V e^{\theta t}
    +\frac{1}{2m} \theta e^{-\theta t} p^{2}
    - e^{-\theta t} \frac{d}{dt} \Big(\frac{p^{2}}{2m} \Big)\\
    &=\frac{1}{2m} \theta  e^{-\theta t} p^{2} 
    + \theta e^{\theta t} V - 2 \dot{p} \frac{p}{m} e^{-\theta t}\\
    & \rightarrow 
    e^{\theta t} \frac{d}{dt}V  
    - e^{-\theta t} \frac{d}{dt} \Big(\frac{p^{2}}{2m} \Big)
    =  - 2 \dot{p} \frac{p}{m} e^{-\theta t}\\
    & \rightarrow 
    e^{\theta t} \frac{d}{dt}V  
    + e^{-\theta t} \frac{d}{dt} \Big(\frac{p^{2}}{2m} \Big)
    =  0\\
    &\rightarrow 
    e^{\theta t} \frac{d}{dt}V  
    + \frac{1}{2} m e^{-\theta t} \frac{d}{dt} \Big(\dot{x}^{2} e^{2 \theta t} \Big)
    =  0\\
    &\rightarrow 
    e^{\theta t} \frac{d}{dt} \Big(
    V  + \frac{1}{2} m \dot{x}^{2} + m \theta \int_{0}^{t} \dot{x}^{2} ds \Big)
    \Big) = 0.
\end{align*}
This leads in fact to the classic statement on energy conservation in Mechanics which includes dissipation. 

\end{document}